\documentclass[11pt]{article}
\usepackage[utf8]{inputenc}
\usepackage{sty}
\usepackage{fullpage}
\usepackage[margin=1in]{geometry}
\usepackage{color}

\newcommand{\lra}[1][{ab,cd}]{\stackrel{#1}{\longleftrightarrow}}

\makeatletter
\renewcommand*{\@fnsymbol}[1]{\ensuremath{\ifcase#1\or *\or \ddagger\or
    \mathsection\or \mathparagraph\or \|\or **\or \dagger\dagger
    \or \ddagger\ddagger \else\@ctrerr\fi}}
\makeatother

\title{The Kikuchi Hierarchy and Tensor PCA}

\usepackage{authblk}
\author[1]{Alexander S.\ Wein\thanks{Email: \textit{aswein@ucdavis.edu}. Partially supported by an Alfred P.\ Sloan Research Fellowship and NSF CAREER Award CCF-2338091. Most of this work was done at the Courant Institute of Mathematical Sciences at New York University, partially supported by NSF grant DMS-1712730 and by the Simons Collaboration on Algorithms and Geometry.}}
\author[2]{Ahmed El Alaoui\thanks{Email: \textit{elalaoui@cornell.edu}. Partially supported by NSF grant DMS-2450867. Most of this work was done at Stanford University, partially supported by NSF IIS-1741162, and ONR N00014-18-1-2729.}}
\author[3]{Cristopher Moore\thanks{Email: \textit{moore@santafe.edu}. Partially supported by NSF grant BIGDATA-1838251.}}

\affil[1]{Department of Mathematics, University of California, Davis}
\affil[2]{Department of Statistics and Data Science, Cornell University}
\affil[3]{Santa Fe Institute}

\date{}

\begin{document}

\maketitle

\begin{abstract}
For the tensor principal component analysis (tensor PCA) problem, we propose a new hierarchy of increasingly powerful algorithms with increasing runtime. Our hierarchy is analogous to the sum-of-squares (SOS) hierarchy but is instead inspired by statistical physics and related algorithms such as belief propagation and AMP (approximate message passing). Our level-$\ell$ algorithm can be thought of as a linearized message-passing algorithm that keeps track of $\ell$-wise dependencies among the hidden variables. Specifically, our algorithms are spectral methods based on the \emph{Kikuchi Hessian}, which generalizes the well-studied Bethe Hessian to the higher-order Kikuchi free energies.

It is known that AMP, the flagship algorithm of statistical physics, has substantially worse performance than SOS for tensor PCA. In this work we `redeem' the statistical physics approach by showing that our hierarchy gives a polynomial-time algorithm matching the performance of SOS. Our hierarchy also yields a continuum of subexponential-time algorithms, and we prove that these achieve the same (conjecturally optimal) tradeoff between runtime and statistical power as SOS. Our proofs are much simpler than prior work, and also apply to the related problem of refuting random $k$-XOR formulas. The results we present here apply to tensor PCA for tensors of all orders, and to $k$-XOR when $k$ is even.

Our methods suggest a new avenue for systematically obtaining optimal algorithms for Bayesian inference problems, and our results constitute a step toward unifying the statistical physics and sum-of-squares approaches to algorithm design.
\end{abstract}

\newpage

\section{Introduction}

High-dimensional Bayesian inference problems are widely studied, including planted clique \cite{jerrum-clique,alon-clique}, sparse PCA \cite{JL-sparse}, and community detection \cite{sbm1,sbm2}, just to name a few. For these types of problems, two general strategies, or \emph{meta-algorithms}, have emerged. The first is rooted in statistical physics and includes the belief propagation (BP) algorithm \cite{pearl-bp,GBP} along with variants such as approximate message passing (AMP) \cite{amp}, and related spectral methods such as linearized BP \cite{spectral-redemption,nb-spectrum}, and the Bethe Hessian \cite{bethe-hessian}. The second meta-algorithm is the \emph{sum-of-squares (SOS) hierarchy} \cite{shor-sos,par-sos,las-sos}, a hierarchy of increasingly powerful semidefinite programming relaxations to polynomial optimization problems, along with spectral methods inspired by it \cite{HSS,HSSS}. Both of these meta-algorithms are known to achieve statistically-optimal performance for many problems. Furthermore, when they fail to perform a task, this is often seen as evidence that no polynomial-time algorithm can succeed. Such reasoning takes the form of free energy barriers in statistical physics~\cite{LKZ-mmse,LKZ-pca} or SOS lower bounds (e.g., \cite{sos-clique}). Thus, we generally expect both meta-algorithms to have optimal statistical performance among all \emph{computationally efficient} algorithms.

A fundamental question is whether we can unify statistical physics and SOS, showing that the two approaches yield, or at least predict, the same performance on a large class of problems. However, one barrier to this comes from the \emph{tensor principal component analysis} (tensor PCA) problem~\cite{RM-tensor}, on which the two meta-algorithms seem to have very different performance. For an integer $p \ge 2$, in the order-$p$ tensor PCA or \emph{spiked tensor} problem we observe a $p$-fold $n \times n \times \cdots \times n$ tensor
\[\bY = \lambda \bx_*^{\otimes p} + \bG\]
where the parameter $\lambda \ge 0$ is a signal-to-noise ratio (SNR), $\bx_* \in \RR^n$ is a planted signal vector with normalization $\|\bx_*\| = \sqrt{n}$ drawn from a simple prior such as the uniform distribution on $\{\pm 1\}^n$, and $\bG$ is a symmetric noise tensor with $\sN(0,1)$ entries. Information-theoretically, it is possible to recover $\bx_*$ given $\bY$ (in the limit $n \to \infty$, with $p$ fixed) when $\lambda \gg n^{(1-p)/2}$ \cite{RM-tensor,tensor-stat}. (Here we ignore log factors, so $A \gg B$ can be understood to mean $A \ge B\, \mathrm{polylog}(n)$.) However, this information-theoretic threshold corresponds to exhaustive search. We would also like to understand the \emph{computational} threshold, i.e., for what values of $\lambda$ there is an efficient algorithm.

The sum-of-squares hierarchy gives a polynomial-time algorithm to recover $\bx_*$ when $\lambda \gg n^{-p/4}$ \cite{HSS}, and SOS lower bounds suggest that no polynomial-time algorithm can do better \cite{HSS,sos-hidden}. However, AMP, the flagship algorithm of statistical physics, is suboptimal for $p \ge 3$ and fails unless $\lambda \gg n^{-1/2}$ \cite{RM-tensor}. Various other ``local'' algorithms such as the tensor power method, Langevin dynamics, and gradient descent also fail below this ``local'' threshold $\lambda \sim n^{-1/2}$ \cite{RM-tensor,alg-thresh}. This casts serious doubts on the optimality of the statistical physics approach.

In this paper we resolve this discrepancy and ``redeem'' the statistical physics approach. The \emph{Bethe free energy} associated with AMP is merely the first level of a hierarchy of \emph{Kikuchi free energies} \cite{kikuchi1,kikuchi2,GBP}. From these Kikuchi free energies, we derive a hierarchy of increasingly powerful algorithms for tensor PCA, similar in spirit to generalized belief propagation \cite{GBP}. Roughly speaking, our level-$\ell$ algorithm can be thought of as an iterative message-passing algorithm that reasons about $\ell$-wise dependencies among the hidden variables. As a result, it has time and space complexity $n^{O(\ell)}$. Specifically, the level-$\ell$ algorithm is a spectral method on a $n^{O(\ell)} \times n^{O(\ell)}$ submatrix of (a first-order approximation of) the \emph{Kikuchi Hessian}, i.e., the matrix of second derivatives of the Kikuchi free energy. This generalizes the \emph{Bethe Hessian} spectral method, which has been successful in the setting of community detection~\cite{bethe-hessian}. We note that the Ph.D.\ dissertation of Saade~\cite{saade-thesis} proposed the Kikuchi Hessian as a direction for future research.

For order-$p$ tensor PCA with $p$ even, we show that level $\ell=p/2$ of the Kikuchi hierarchy gives an algorithm that succeeds down to the SOS threshold $\lambda \sim n^{-p/4}$, closing the gap between SOS and statistical physics. Furthermore, by taking $\ell = n^\delta$ levels for various values of $\delta \in (0,1)$, we obtain a continuum of subexponential-time algorithms that achieve a precise tradeoff between runtime and the signal-to-noise ratio---exactly the same tradeoff curve that SOS is known to achieve~\cite{mult-approx,cert-tensor}.\footnote{The strongest SOS results only apply to a variant of the spiked tensor model with Rademacher observations, but we do not expect this difference to be important; see Section~\ref{sec:subexp}.} We obtain similar results when $p$ is odd, by combining a matrix related to the Kikuchi Hessian with a construction similar to~\cite{coja-sat}; see Section~\ref{sec:cert-odd}.

Our approach also applies to the problem of refuting random $k$-XOR formulas when $k$ is even, showing that we can strongly refute random formulas with $n$ variables and $m \gg n^{k/2}$ clauses in polynomial time, and with a continuum of subexponential-time algorithms that succeed at lower densities. This gives a much simpler proof of the results of~\cite{strongly-refuting}, using only the matrix Chernoff bound instead of intensive moment calculations; see Section~\ref{sec:xor}.

Our results redeem the statistical physics approach to algorithm design and give hope that the Kikuchi hierarchy provides a systematic way to derive optimal algorithms for a large class of Bayesian inference problems. We see this as a step toward unifying the statistical physics and SOS approaches. Indeed, we propose the following informal meta-conjecture: for high-dimensional inference problems with planted solutions (and related problems such as refuting random constraint satisfaction problems) the SOS hierarchy and the Kikuchi hierarchy both achieve the optimal tradeoff between runtime and statistical power.

\paragraph{Concurrent and subsequent work.}

After the initial appearance of this paper, some related independent work has appeared. A hierarchy of algorithms similar to ours is proposed by~\cite{hastings-quantum}, but with a different motivation based on a system of quantum particles. For the problem of refuting random $k$-XOR formulas with $k$ even, the work of~\cite{ahn-simpler} gives, like us, a simplification of~\cite{strongly-refuting}, but using a different approach. Also, \cite{replicated} gives an alternative ``redemption'' of local algorithms for tensor PCA, based on ``averaged gradient descent.''

Furthermore, the conference version of the present paper~\cite{focs-version} led to a number of significant follow-up works. The ``Kikuchi matrices'' that we introduced have played a key role in state-of-the-art results on CSP refutation (in both random and smoothed settings)~\cite{smoothed-random}, locally decodable/correctable codes~\cite{ldc,lcc-1,lcc-2}, as well as the resolution, up to a logarithmic factor, of Feige's conjecture~\cite{feige-conj} on the hypergraph Moore bound in extremal combinatorics~\cite{smoothed-random,simpler-moore}. Notably, for $k$-XOR refutation, the work of~\cite{smoothed-random} generalized our approach from even $k$ to odd $k$, which required a more involved algorithm and analysis than we initially expected. Some of the above results use the Kikuchi matrices in conjunction with SOS, as these matrices constitute spectral SOS certificates for e.g.\ CSP refutation.

\section{Preliminaries and Prior Work}

\subsection{Notation}

Our asymptotic notation (e.g., $O(\cdot), o(\cdot), \Omega(\cdot),\omega(\cdot)$) pertains to the limit $n \to \infty$ (large dimension) and may hide constants depending on $p$ (tensor order), which we think of as fixed. We say an event occurs \emph{with high probability} if it occurs with probability $1-o(1)$.

A \emph{tensor} $\bT \in (\RR^n)^{\otimes p}$ is an $n \times n \times \cdots \times n$ ($p$ times) multi-array with entries denoted by $\bT_{i_1,\ldots,i_p}$, where $i_k \in [n] \defeq \{1,2,\ldots,n\}$. We call $p$ the \emph{order} of $\bT$ and $n$ the \emph{dimension} of $\bT$. For a vector $\bu \in \RR^n$, the rank-1 tensor $\bu^{\otimes p}$ is defined by $(\bu^{\otimes p})_{i_1,\ldots,i_p} = \prod_{k=1}^p \bu_{i_k}$. A tensor $\bT$ is \emph{symmetric} if $\bT_{i_1,\ldots,i_p} = \bT_{i_{\pi(1)},\ldots,i_{\pi(p)}}$ for any permutation $\pi \in S_p$. For a symmetric tensor, if $E=\{i_1,\ldots,i_p\} \subseteq [n]$, we will often write $\bT_E \defeq \bT_{i_1,\ldots,i_p}$.

\subsection{The Spiked Tensor Model}

A general formulation of the spiked tensor model is as follows. For an integer $p \ge 2$, let $\widetilde{\bG} \in (\RR^n)^{\otimes p}$ be an order-$p$ tensor with entries sampled i.i.d.\ from $\sN(0,1)$. We then symmetrize this tensor and obtain a symmetric tensor $\bG$, 
\[\bG := \frac{1}{\sqrt{p!}}\sum_{\pi \in S_p} \widetilde{\bG}^\pi,\]
where $S_p$ is the symmetric group of permutations of $[p]$, and $\widetilde{\bG}^\pi_{i_1,\ldots,i_p} := \widetilde{\bG}_{i_{\pi(1)},\ldots,i_{\pi(p)}}$. Note that if $i_1,\ldots,i_p$ are distinct then $\bG_{i_1,\ldots,i_p} \sim \sN(0,1)$. We draw a `signal' vector or `spike' $\bx_* \in \RR^n$ from a prior distribution $P_{\tx}$ supported on the sphere $\mathcal{S}^{n-1}(\sqrt{n}) = \{\bx \in \RR^n \,:\, \|\bx\| = \sqrt{n}\}$. Then we let $\bY \in (\RR^n)^{\otimes p}$ be the tensor 
\begin{equation}\label{spiked_tensor}
\bY = \lambda \bx_*^{\otimes p} + \bG \, .
\end{equation}
We will mostly focus on the \emph{Rademacher-spiked model} where $\bx_*$ is uniform in $\{ \pm 1 \}^n$, i.e.,  
$P_{\tx} = 2^{-n} \prod_i  \big( \delta_{-1}(x_i)+\delta_{+1}(x_i) \big)$.
We will sometimes state results without specifying the prior $P_{\tx}$, in which case the result holds for \emph{any} prior normalized so that $\|\bx_*\| = \sqrt{n}$. Let $\PP_\lambda$ denote the law of the tensor $\bY$.  The parameter $\lambda = \lambda(n)$ may depend on $n$. We will consider the limit $n \to \infty$ with $p$ held fixed.

Our algorithms will depend only on the entries $\bY_{i_1,\ldots,i_p}$ where the indices $i_1,\ldots,i_p$ are distinct: that is, on the collection
\[
\left\{ \bY_E = \lambda \bx_*^E + \bG_E \,:\, E \subseteq [n], |E|=p \right\} \, , 
\]
where $\bG_E \sim \mathcal{N}(0,1)$ and for a vector $\bx \in \RR^{n}$ we write $\bx^E = \prod_{i \in E} \bx_i$.

Perhaps one of the simplest statistical tasks is binary hypothesis testing. In our case this amounts to, given a tensor $\bY$ as input with the promise that it was sampled from $\PP_{\lambda}$ with $\lambda \in \{0,\overline{\lambda}\}$, determining whether $\lambda = 0$ or $\lambda = \overline{\lambda}$. We refer to $\PP_{\lambda}$ for $\lambda>0$ as the \emph{planted} distribution, and $\PP_{0}$ as the \emph{null} distribution.  

\begin{definition}
We say that an algorithm (or test) $f: (\RR^n)^{\otimes p} \to \{0,1\}$ achieves \emph{strong detection} between $\PP_{0}$ and $\PP_{\lambda}$ if
\[ \lim_{n \to \infty} \PP_{\lambda}(f(\bY) = 1) = 1 \qquad \text{and} \qquad \lim_{n \to \infty} \PP_{0}(f(\bY) = 0) = 1.\]
Additionally we say that $f$ achieves \emph{weak detection} between $\PP_{0}$ and $\PP_{\lambda}$ if the sum of Type-I and Type-II errors remains strictly below 1:
\[\limsup_{n \to \infty} \big\{\PP_{0}(f(\bY) = 1) + \PP_{\lambda}(f(\bY) = 0)\big\} < 1.\]
\end{definition}

An additional goal is to recover the planted vector $\bx_*$. Note that when $p$ is even, $\bx_*$ and $-\bx_*$ have the same posterior probability. Thus, our goal is to recover $\bx_*$ up to a sign.

\begin{definition}
The \emph{normalized correlation} between vectors $\hat\bx, \bx \in \RR^n$ is
$$\corr(\hat\bx,\bx) = \frac{|\langle \hat \bx, \bx \rangle|}{\|\hat \bx\| \|\bx\|}.$$
\end{definition}

\begin{definition}
An estimator $\hat \bx = \hat \bx(\bY)$ achieves \emph{weak recovery} if $\corr(\hat \bx,\bx_*)$ is lower-bounded by a strictly positive constant---and we write $\corr(\hat \bx,\bx_*) = \Omega(1)$---with high probability, and achieves \emph{strong recovery} if $\corr(\hat \bx,\bx_*) = 1-o(1)$ with high probability.
\end{definition}

\noindent We expect that strong detection and weak recovery are generally equally difficult, although formal implications are not known in either direction. We will see in Section~\ref{sec:boost} that in some regimes, weak recovery and strong recovery are equivalent.

\paragraph{The matrix case.}

When $p = 2$, the spiked tensor model reduces to the \emph{spiked Wigner model}. We know from random matrix theory that when $\lambda = \hat\lambda/\sqrt{n}$ with $\hat\lambda > 1$, strong detection is possible by thresholding the maximum eigenvalue of $\bY$, and weak recovery is achieved by PCA, i.e., taking the leading eigenvector~\cite{FP,pca-eigenvec}. For many spike priors including Rademacher, strong detection and weak recovery are statistically impossible when $\hat\lambda < 1$~\cite{MRZ-sphere,DAM,PWBM-contig}, so there is a sharp phase transition at $\hat\lambda = 1$. 
(Note that weak detection is still possible below $\hat\lambda = 1$~\cite{alaoui2018fundamental}.) A more sophisticated algorithm is AMP (approximate message passing) \cite{amp,FR-amp,LKZ-mmse,DAM}, which can be thought of as a modification of the matrix power method which uses certain nonlinear transformations to exploit the structure of the spike prior. For many spike priors including Rademacher, AMP is known to achieve the information-theoretically optimal correlation with the true spike \cite{DAM,replica-proof}. However, like PCA, AMP achieves zero correlation asymptotically when $\hat\lambda < 1$. For certain spike priors (e.g., sparse priors), statistical-computational gaps can appear in which it is information-theoretically possible to succeed for some $\hat\lambda < 1$ but we do not expect that polynomial-time algorithms can do so \cite{LKZ-mmse,LKZ-pca,BMVVX}.

\paragraph{The tensor case.}

The tensor case $p \ge 3$ was introduced by \cite{RM-tensor}, who also proposed various algorithms. Information-theoretically, there is a sharp phase transition similar to the matrix case: for many spike priors $P_{\tx}$ including Rademacher, if $\lambda = \hat\lambda n^{(1-p)/2}$ then weak recovery is possible when $\hat\lambda > \lambda_c$ and impossible when $\hat\lambda < \lambda_c$ for a particular constant $\lambda_c = \lambda_c(p,P_{\tx})$ depending on $p$ and $P_{\tx}$~\cite{tensor-stat}. 
Strong detection undergoes the same transition, being possible if $\lambda > \lambda_c$ and impossible otherwise~\cite{chen2017phase,chen2018phase,stat-thresh} (see also \cite{PWB-tensor}). In fact, it is shown in these works that even weak detection is impossible below $\lambda_c$, in sharp contrast with the matrix case.
There are polynomial-time algorithms (e.g., SOS) that succeed at both strong detection and strong recovery when $\lambda \gg n^{-p/4}$ for any spike prior \cite{RM-tensor,HSS,HSSS}, which is above the information-theoretic threshold by a factor that diverges with $n$. There are also SOS lower bounds suggesting that (for many priors) no polynomial-time algorithm can succeed at strong detection or weak recovery when $\lambda \ll n^{-p/4}$ \cite{HSS,sos-hidden}. Thus, unlike the matrix case, we expect a large statistical-computational gap, i.e., a ``possible-but-hard'' regime when $n^{(1-p)/2} \ll \lambda \ll n^{-p/4}$. Thresholds for tensor PCA are depicted in Figure~\ref{fig:thresholds} and discussed further in the following subsections.

\begin{figure}
    \centering
    \begin{tikzpicture}[xscale=2]
        \def \d{.2}; 
        \def \y{.6}; 
        \draw[ultra thick, -Stealth] (0,0) -- (4,0);
        \draw[ultra thick] (1,-1*\d) -- (1,\d);
        \draw[ultra thick] (2,-1*\d) -- (2,\d);
        \draw[ultra thick] (3,-1*\d) -- (3,\d);
        \node at (1,-1*\y) {$n^{(1-p)/2}$};
        \node at (2,-1*\y) {$n^{-p/4}$};
        \node at (3,-1*\y) {$n^{-1/2}$};
        \node at (4,-1*\y) {$\lambda$};
        \node at (1, \y) {IT};
        \node at (2, \y) {SOS};
        \node at (3, \y) {AMP};
    \end{tikzpicture}
    \caption{Thresholds for tensor PCA: (IT) Information-theoretic threshold; (SOS) Threshold for the best known poly-time algorithms, including sum-of-squares and certain spectral methods; (AMP) Threshold for approximate message passing and other ``local'' algorithms, which are suboptimal.}
    \label{fig:thresholds}
\end{figure}

\subsection{The Tensor Power Method and Local Algorithms}
\label{sec:local}

Various algorithms have been proposed and analyzed for tensor PCA \cite{RM-tensor,HSS,HSSS,homotopy,alg-thresh}. We will present two such algorithms that are simple, representative, and will be relevant to the discussion. The first is the \emph{tensor power method} \cite{AGHKT-power,RM-tensor,AGJ-power}.

\begin{algorithm}(Tensor Power Method)
For a vector $\bu \in \RR^n$ and a tensor $\bY \in (\RR^n)^{\otimes p}$, let $\bY\{\bu\} \in \RR^n$ denote the vector with entries 
\[(\bY\{\bu\})_i = \sum_{j_1,\ldots,j_{p-1}} \bY_{i,j_1,\ldots,j_{p-1}} \bu_{j_1} \cdots \bu_{j_{p-1}}, \quad\quad i \in [n].\]
The \emph{tensor power method} begins with an initial guess $\bu \in \RR^n$ (e.g., chosen at random) and repeatedly iterates the update rule $\bu \leftarrow \bY\{\bu\}$ until $\bu/\|\bu\|$ converges.
\end{algorithm}

\noindent The tensor power method appears to only succeed when $\lambda \gg n^{-1/2}$ \cite{RM-tensor}, which is worse than the SOS threshold $\lambda \sim n^{-p/4}$. The AMP algorithm of \cite{RM-tensor} is a more sophisticated variant of the tensor power method, but AMP also fails unless $\lambda \gg n^{-1/2}$ \cite{RM-tensor}. Two other related algorithms, gradient descent and Langevin dynamics, also fail unless $\lambda \gg n^{-1/2}$ \cite{alg-thresh}. Following \cite{alg-thresh}, we refer to all of these algorithms (tensor power method, AMP, gradient descent, Langevin dynamics) as \emph{local algorithms}, and we refer to the corresponding threshold $\lambda \sim n^{-1/2}$ as the \emph{local threshold}. Here ``local'' is not a precise notion, but roughly speaking, local algorithms keep track of a current guess for $\bx_*$ and iteratively update it to a nearby vector that is more favorable in terms of e.g. the log-likelihood.
This discrepancy between local algorithms and SOS is what motivated the current work.

\subsection{The Tensor Unfolding Algorithm}
\label{sec:unfolding}

We have seen that local algorithms do not seem able to reach the SOS threshold. Let us now describe one of the simplest algorithms that does reach this threshold: \emph{tensor unfolding}. Tensor unfolding was first proposed by \cite{RM-tensor}, where it was shown to succeed when $\lambda \gg n^{-\lfloor p/2 \rfloor/2}$ and conjectured to succeed when $\lambda \gg n^{-p/4}$ (the SOS threshold). For the case $p=3$, the same algorithm was later reinterpreted as a spectral relaxation of SOS, and proven to succeed\footnote{The analysis of \cite{HSS} applies to a close variant of the spiked tensor model in which the noise tensor is asymmetric. We do not expect this difference to be important.} when $\lambda \gg n^{-3/4} = n^{-p/4}$ \cite{HSS}, confirming the conjecture of \cite{RM-tensor}. We now present the tensor unfolding method, restricting to the case $p=3$ for simplicity. There is a natural extension to all $p$ \cite{RM-tensor}, and (a close variant of) this algorithm will in fact appear as level $\ell = \lfloor p/2 \rfloor$ in our hierarchy of algorithms (see Section~\ref{sec:main-results} and Appendix~\ref{app:odd}).

\begin{algorithm}(Tensor Unfolding)
Given an order-3 tensor $\bY \in (\RR^n)^{\otimes 3}$, flatten it to an $n \times n^2$ matrix $\bM$, i.e., let $\bM_{i,jk} = \bY_{ijk}$. Compute the leading eigenvector of $\bM \bM^\top$.
\end{algorithm}

If we use the matrix power method to compute the leading eigenvector, we can restate the tensor unfolding method as an iterative algorithm: keep track of state vectors $\bu \in \RR^n$ and $\bv \in \RR^{n^2}$, initialize $\bu$ randomly, and alternate between applying the update steps $\bv \leftarrow \bM^\top \bu$ and $\bu \leftarrow \bM \bv$. We will see later (see Section~\ref{sec:intuition}) that this can be interpreted as a message-passing algorithm between singleton indices, represented by $\bu$, and \emph{pairs} of indices, represented by $\bv$. Thus, tensor unfolding is not ``local'' in the sense of Section~\ref{sec:local} because it keeps a state of size $O(n^2)$ (keeping track of \emph{pairwise} information) instead of size $O(n)$. We can, however, think of it as local on a ``lifted'' space, and this allows it to surpass the local threshold.

Other methods have also been shown to achieve the SOS threshold $\lambda \sim n^{-p/4}$, including SOS itself and various spectral methods inspired by it \cite{HSS,HSSS}.

\subsection{Boosting and Linearization}
\label{sec:boost}

One fundamental difference between the matrix case ($p=2$) and tensor case ($p \ge 3$) arise from the following boosting property. The following result, implicit in~\cite{RM-tensor}, shows that if $\lambda$ is substantially above the information-theoretic threshold (i.e., $\lambda \gg n^{(1-p)/2}$) then weak recovery can be boosted to strong recovery via a single power iteration. We give a proof in Appendix~\ref{app:boost}.
\begin{proposition}\label{prop:boost}
Let $\bY \sim \PP_{\lambda}$ with any spike prior $P_{\tx}$ supported on $\mathcal{S}^{n-1}(\sqrt{n})$. Suppose we have an initial guess $\bu \in \RR^n$ satisfying $\corr(\bu,\bx_*) \ge \tau$. Obtain $\widehat \bx$ from $\bu$ via a single iteration of the tensor power method: $\widehat \bx = \bY\{\bu\}$. There exists a constant $c = c(p)>0$ such that with high probability, 
\[\corr(\widehat \bx,\bx_*) \ge 1 - c\lambda^{-1} \tau^{1-p} n^{(1-p)/2}.\] 
In particular, if $\tau > 0$ is any constant and $\lambda = \omega(n^{(1-p)/2})$ then $\corr(\widehat \bx,\bx) = 1-o(1)$.
\end{proposition}

\noindent For $p \ge 3$, since we do not expect polynomial-time algorithms to succeed when $\lambda = O(n^{(1-p)/2})$, this implies an ``all-or-nothing'' phenomenon: for a given $\lambda = \lambda(n)$, the optimal polynomial-time algorithm will either achieve correlation that is asymptotically 0 or asymptotically 1. This is in stark contrast to the matrix case where, for $\lambda = \hat\lambda/\sqrt{n}$, the optimal correlation is a constant (in $[0,1]$) depending on both $\hat\lambda$ and the spike prior $P_{\tx}$. This contrast is not due to Proposition~\ref{prop:boost} itself---which applies equally well to $p=2$---but due to the large statistical-computational gap present in the tensor case, as this means the regime of interest lies well above the information-theoretic threshold.

This boosting result substantially simplifies things when $p \ge 3$ because it implies that the only important question is identifying the threshold for weak recovery, instead of trying to achieve the optimal correlation. Heuristically, since we only want to attain the optimal threshold, statistical physics suggests that we can use a simple ``linearized'' spectral algorithm instead of a more sophisticated nonlinear algorithm. To illustrate this in the matrix case ($p=2$), one needs to use AMP in order to achieve optimal correlation, but one can achieve the optimal threshold using linearized AMP, which boils down to computing the top eigenvector. In the related setting of community detection in the stochastic block model, one needs to use belief propagation to achieve optimal correlation~\cite{sbm1,sbm2,MNS-bp}, but one can achieve the optimal threshold using a linearized version of belief propagation, which is a spectral method on the non-backtracking walk matrix~\cite{spectral-redemption,nb-spectrum} or the related \emph{Bethe Hessian}~\cite{bethe-hessian}. Our spectral methods for tensor PCA are based on the \emph{Kikuchi Hessian}, which is a generalization of the Bethe Hessian.

\subsection{Subexponential-time Algorithms for Tensor PCA}\label{sec:subexp}

The degree-$\ell$ sum-of-squares algorithm is a large semidefinite program that requires runtime $n^{O(\ell)}$ to solve. Oftentimes the regime of interest is when $\ell$ is constant, so that the algorithm runs in polynomial time. However, one can also explore the power of subexponential-time algorithms by letting $\ell = n^\delta$ for $\delta \in (0,1)$, which gives an algorithm with runtime roughly $2^{n^\delta}$. Results of this type are known for tensor PCA \cite{strongly-refuting,mult-approx,cert-tensor}. The strongest such results are for a different variant of tensor PCA, which we now define.

\begin{definition}
In the \emph{order-$p$ discrete spiked tensor model} with spike prior $P_{\tx}$ supported on $\mathcal{S}^{n-1}(\sqrt{n})$, and SNR parameter $\lambda \ge 0$, we draw a spike $\bx_* \sim P_{\tx}$ and then for each $1 \le i_1 \le \cdots \le i_p \le n$, we independently observe a $\{\pm 1\}$-valued random variable $\bY_{i_1,\ldots,i_p}$ with $\EE[\bY_{i_1,\ldots,i_p}] = \lambda (\bx_*^{\otimes p})_{i_1,\ldots,i_p}$.
\end{definition}
\noindent This model differs from our usual one in that the observations are conditionally Rademacher instead of Gaussian, but we do not believe this makes an important difference. However, for technical reasons, the known SOS results are strongest in this discrete setting.

\begin{theorem}[\cite{cert-tensor,HSS}]\label{SOS-guarantee}
For any $1 \le \ell \le n$, there is an algorithm with runtime $n^{O(\ell)}$ that achieves strong detection and strong recovery in the order-$p$ discrete spiked tensor model (with any spike prior) whenever
$$\lambda \ge \ell^{1/2 - p/4}n^{-p/4} \,\mathrm{polylog}(n).$$
\end{theorem}
\noindent The work of \cite{cert-tensor} shows how to certify an upper bound on the injective norm of a random $\{\pm 1\}$-valued tensor, which immediately implies the algorithm for strong detection. When combined with \cite{HSS}, this can also be made into an algorithm for strong recovery (see Lemma~4.4 of \cite{HSS}). Similar (but weaker) SOS results are also known for the standard spiked tensor model (see \cite{strongly-refuting} and \emph{arXiv} version 1 of \cite{cert-tensor}), and we expect that Theorem~\ref{SOS-guarantee} also holds for this case. (Curiously, the certification-to-recovery reduction within SOS is not known in the related setting of CSPs.)

When $\ell = n^\delta$ for $\delta \in (0,1)$, Theorem~\ref{SOS-guarantee} implies that we have an algorithm with runtime $n^{O(n^\delta)} = 2^{n^{\delta+o(1)}}$ that succeeds when $\lambda \gg n^{\delta/2 - p \delta /4  - p/4}$. Note that this interpolates smoothly between the polynomial-time threshold ($\lambda \sim n^{-p/4}$) when $\delta = 0$, and the information-theoretic threshold ($\lambda \sim n^{(1-p)/2}$) when $\delta = 1$. We will prove (for $p$ even) that our algorithms achieve this same tradeoff, and we expect this tradeoff to be optimal.\footnote{One form of evidence suggesting that this tradeoff is optimal is based on the \emph{low-degree likelihood ratio}; see~\cite{low-deg-notes}. There is also an SOS lower bound~\cite{potechin2022sub}.}

\section{Main results} \label{sec:main-results}

In this section we present our main results about detection and recovery in the spiked tensor model. We propose a hierarchy of spectral methods, which are directly derived from the hierarchy of \emph{Kikuchi free energies}. Specifically, the symmetric difference matrix defined below appears (approximately) as a submatrix of the Hessian of the Kikuchi free energy. The full details of this derivation are given in Section~\ref{sec:motivation} and Appendix~\ref{app:hessian}. For now we simply state the algorithms and results.

We will restrict our attention to the Rademacher-spiked tensor model, which is the setting in which we derived our algorithms. However, we show in Appendix~\ref{app:all-priors} that the same algorithm works for a large class of priors (at least for strong detection). Furthermore, we show in Section~\ref{sec:xor} that the same algorithm can also be used for refuting random $k$-XOR formulas (when $k$ is even).

We will also restrict to the case where the tensor order $p$ is even. The case of odd $p$ is discussed in Appendix~\ref{app:odd}, where we give an algorithm derived from the Kikuchi Hessian and conjecture that it achieves optimal performance. We are unable to prove this, but we are able to prove that optimal results are attained by a related algorithm (see Section~\ref{sec:cert-odd}).

Our approach requires the introduction of two matrices:
\paragraph{The symmetric difference matrix of order $\ell$.} Let $p$ be even and let $\bY \in (\RR^n)^{\otimes p}$ be the observed order-$p$ symmetric tensor. We will only use the entries $\bY_{i_1,\ldots,i_p}$ for which the indices $i_1,\ldots,i_p$ are distinct; we denote such entries by $\bY_E$ where $E \subseteq [n]$ with $|E| = p$. Fix an integer $\ell \in [p/2,n-p/2]$, and consider the symmetric ${n \choose \ell} \times {n \choose \ell}$ matrix $\bM$ indexed by sets $S \subseteq [n]$ of size $\ell$, having entries 
\begin{equation}\label{symmetric_diff}
\bM_{S,T} = 
\begin{cases}
\bY_{S \symd T}  & \mbox{if } |S \symd T| = p,\\
0 & \mbox{otherwise}.
\end{cases}
\end{equation}       
Here $\symd$ denotes the symmetric difference between sets. The leading eigenvector of $\bM$ is intended to be an estimate of $(\bx_*^S)_{|S|=\ell}$ where $\bx^S \defeq \prod_{i \in S} \bx_i$. The following `voting' matrix is a natural rounding scheme to extract an estimate of $\bx_*$ from such a vector.

\paragraph{The voting matrix.} To a vector $\bv \in \RR^{{n\choose \ell}}$ we associate the following symmetric $n \times n$ voting matrix $\bV(\bv)$ having entries 
\begin{equation}\label{voting_matrix}
\bV_{ii}(\bv) = 0 \;\;\forall i \in [n], 
\quad \mbox{and} \quad 
\bV_{ij}(\bv) = \frac{1}{2}\sum_{S,T \in  {[n] \choose \ell}} \bv_S \bv_T \bone_{S \symd T = \{i,j\}} \;\;\forall i\neq j.
\end{equation}
Let us define the important quantity 
\begin{equation}
\label{eq:dell}
d_\ell \defeq {n-\ell \choose p/2}{\ell \choose p/2} \, . 
\end{equation}   
This is the number of sets $T$ of size $\ell$ such that $|S \symd T|=p$ for a given set $S$ of size $\ell$. (Note it is important for $p$ to be even here.) Now we are in position to formulate our algorithms for detection and recovery. 

\begin{algorithm}[Detection for even $p$]\label{detection_even}
~~
\begin{enumerate}
\item Compute the top eigenvalue $\lambda_{\max}(\bM)$ of the symmetric difference matrix $\bM$. 
\item Reject the null hypothesis $\PP_{0}$ (i.e., return `1') if $\lambda_{\max}(\bM) \ge  \lambda d_\ell/2$. 
\end{enumerate}
\end{algorithm}
 We will see in Section~\ref{sec:symm_diff} that the quantity $\lambda d_\ell/2$ is half the largest eigenvalue of `the signal part' of the matrix $\bM$. This intuitively justifies its appearance in the above algorithm.

\begin{algorithm}[Recovery for even $p$]\label{recovery_even}
~~
\begin{enumerate}
\item Compute a (unit-norm) leading eigenvector\footnote{We define a leading eigenvector to be an eigenvector whose eigenvalue is maximal (although our results still hold for an eigenvector whose eigenvalue is maximal in absolute value).} $\bv_{\textup{top}} \in \RR^{{n\choose \ell}}$ of $\bM$.
\item Form the associated voting matrix $\bV(\bv_{\textup{top}})$.
\item Compute a leading eigenvector $\widehat{\bx}$ of $\bV(\bv_{\textup{top}})$, and output $\widehat{\bx}$. 
\end{enumerate}

\end{algorithm}

The next two theorems characterize the performance of Algorithms~\ref{detection_even} and~\ref{recovery_even} for the strong detection and recovery tasks, respectively. The proofs can be found in Section~\ref{sec:symm_diff}.
\begin{theorem}\label{detection_result}
Consider the Rademacher-spiked tensor model with $p$ even. For all $\lambda \ge 0$ and $\ell \in [p/2,n-p/2]$, we have
\[\PP_{0} \Big(\lambda_{\max}(\bM) \ge \frac{\lambda d_\ell}{2}\Big) \vee 
\PP_{\lambda} \Big(\lambda_{\max}(\bM) \le \frac{\lambda d_\ell}{2}\Big) \le 2n^{\ell}e^{-\lambda^2 d_\ell/8}.\]
(Here, $a \vee b = \max\{a,b\}$.) Therefore, Algorithm~\ref{detection_even} achieves strong detection between $\PP_{0}$ and $\PP_{\lambda}$  if $\frac{1}{8} \lambda^2 d_\ell - \ell \log n \to +\infty$ as $n \to +\infty$. 
\end{theorem} 

\begin{theorem}\label{recovery_result}
Consider the Rademacher-spiked tensor model with $p$ even. Let $\widehat{\bx} \in \RR^n$ be the output of Algorithm~\ref{recovery_even}. There exists an absolute constant $c_0>0$ such that for all $\epsilon >0$ and $\delta \in (0,1)$, if $\ell \le n \epsilon^2$ and $\lambda \ge c_0\epsilon^{-4} \sqrt{\log (n^{\ell}/\delta) \big/ d_\ell}$, then $\corr(\widehat{\bx},\bx_*) \ge 1-c_0\epsilon$ with probability at least $1-\delta$. 
\end{theorem}

\begin{remark}
If $\ell = o(n)$, we have $d_\ell = \Theta(n^{p/2}\ell^{p/2})$, and so the above theorems imply that strong detection and strong recovery are possible as soon as $\lambda \gg \ell^{- (p-2)/4} n^{-p/4} \sqrt{\log n}$. Comparing with Theorem~\ref{SOS-guarantee}, this scaling coincides with the guarantees achieved by the level-$\ell$ SOS algorithm of~\cite{cert-tensor}, up to a possible discrepancy in logarithmic factors.
\end{remark}

Due to the particularly simple structure of the symmetric difference matrix $\bM$ (in particular, the fact that its entries are simply entries of $\bY$), the proof of detection (Theorem~\ref{detection_result}) follows from a straightforward application of the matrix Chernoff bound. In contrast, the corresponding SOS results \cite{strongly-refuting,mult-approx,cert-tensor} work with more complicated matrices involving high powers of the entries of $\bY$, and the analysis is much more involved.

Our proof of recovery is unusual in that the signal component of $\bM$, call it $\bX$, is not rank-one (see Eq.~\eqref{signal_noise}); it even has a vanishing spectral gap when $\ell \gg 1$. Thus, the leading eigenvector of $\bM$ does not correlate well with the leading eigenvector of $\bX$. While this may seem to render recovery hopeless at first glance, this is not the case, due to the fact that \emph{many} eigenvectors (actually, eigenspaces) of $\bX$ contain non-trivial information about the spike $\bx_*$, as opposed to only the top one. We prove this by exploiting the special structure of $\bX$ through the \emph{Johnson scheme}, and using tools from Fourier analysis on a slice of the hypercube, in particular a Poincar\'e-type inequality by \cite{filmus}.

\paragraph{Removing the logarithmic factor} Both Theorem~\ref{detection_result} and Theorem~\ref{recovery_result} involve a logarithmic factor in $n$ in the lower bound on SNR $\lambda$.
These logarithmic factors are an artifact of the matrix Chernoff bound, and we believe they can be removed. The analysis of \cite{HSS} removes the logarithmic factors for the tensor unfolding algorithm, which is essentially the case $p = 3$ and $\ell = 1$ of our algorithm.
This suggests the following precise conjecture on the power of polynomial-time algorithms.
\begin{conjecture}
Fix $p$ and let $\ell$ be constant (not depending on $n$). There exists a constant $c_p(\ell) > 0$ with $c_p(\ell) \to 0$ as $\ell \to \infty$ (with $p$ fixed) such that if $\lambda \ge c_p(\ell) n^{-p/4}$ then Algorithm~\ref{detection_even} and Algorithm~\ref{recovery_even} (which run in time $n^{O(\ell)}$) achieve strong detection and strong recovery, respectively. Specifically, we expect $c_p(\ell) \sim \ell^{-(p-2)/4}$ for large $\ell$.
\end{conjecture}

\paragraph{Followup work:} The above conjecture was partially settled in the recent work~\cite{bandeira2024matrix} for $\ell \in [p/2,3p/4)$, where the authors show a sharp estimate for the first eigenvalue of the matrix $\bM/\sqrt{d_{\ell}}$ by relating it to a corresponding free probability model which is amenable to sharp analysis; a line of inquiry yielding remarkable results initiated in~\cite{bandeira2023matrix}. In particular, the authors show that the test described in Algorithm~\ref{detection_even} achieves strong detection for $\lambda \ge (1+\epsilon) d_{\ell}^{-1/2}$ when $p/2 \le \ell < 3p/4$ for any $\epsilon>0$, see~\cite[Section 3.3]{bandeira2024matrix}.

\section{Motivating the Symmetric Difference Matrices} \label{sec:motivation}

In this section we motivate the symmetric difference matrices used in our algorithms. In Section~\ref{sec:intuition} we give some high-level intuition. In Section~\ref{sec:var-kik} we give a more principled derivation based on the Kikuchi Hessian, with many of the calculations deferred to Appendix~\ref{app:hessian}.

\subsection{Intuition: Higher-Order Message-Passing and Maximum Likelihood} \label{sec:intuition}

We will give two related justifications for the symmetric difference matrices. First, we will see how our algorithms can be thought of as iterative message-passing procedures among subsets of size $\ell$. Second, we will see that the symmetric difference matrix is a submatrix of some exponential-size matrix whose leading eigenvector can exactly compute the maximum likelihood estimator.

As stated previously, our algorithms will choose to ignore the entries $\bY_{i_1,\ldots,i_p}$ for which $i_1,\ldots,i_p$ are not distinct; these entries turn out to be unimportant asymptotically. We restrict to the Rademacher-spiked tensor model, as this yields a clean and simple derivation. The posterior distribution for the spike $\bx_*$ given the observed tensor $\bY$ is 
\begin{align}
\PP(\bx \mid \bY) &\propto \exp\Big\{-\frac{1}{2} \sum_{i_1 < \cdots < i_p} \big(\bY_{i_1 \cdots i_p} - \lambda \bx_{i_1} \cdots \bx_{i_p}\big)^2\Big\} \nonumber\\
&\propto \exp\Big\{\lambda \sum_{i_1 < \cdots < i_p} \bY_{i_1 \cdots i_p} \bx_{i_1} \cdots \bx_{i_p} \Big\}
= \exp\Big\{\lambda \sum_{|E|=p} \bY_E \bx^E \Big\}
\label{eq:posterior}
\end{align}
over the domain $\bx \in \{\pm 1\}^n$.
\noindent We take $p$ to be even; a similar derivation works for odd $p$. Now fix $\ell \in [p/2,n-p/2]$. We can write the above as
\begin{equation}\label{eq:log-like}
\PP(\bx \mid \bY)  \propto \exp\Big\{\frac{\lambda}{N} \sum_{|S \symd T| = p} \bY_{S \symd T} \bx^S \bx^T\Big\} \, , 
\end{equation}
where the sum is over ordered pairs $(S,T)$ of sets $S,T \subseteq [n]$ with $|S| = |T| = \ell$ and $|S \symd T| = p$, and where $N = d_\ell {n \choose \ell} / {n \choose p}$ is the number of terms $(S,T)$ with a given symmetric difference $E$.

A natural algorithm to maximize the log-likelihood is the following. For each $S \subseteq [n]$ of size $|S| = \ell$, keep track of a variable $\bu_S \in \RR$, which is intended to be an estimate of $\bx_*^S \defeq \prod_{i\in S} (\bx_*)_i$. Note that there are consistency constraints that $(\bx_*^S)_{|S|=\ell}$ must obey, such as $\bx_*^S \bx_*^T \bx_*^V = 1$ when $S \symd T \symd V = \varnothing$; we will relax the problem and will not require our vector $\bu=(\bu_S)_{|S|=\ell}$ to obey such constraints. Instead, we simply attempt to maximize
\begin{equation}\label{eq:umu}
\frac{1}{\|\bu\|^2} \sum_{|S \symd T|=p} \bY_{S \symd T} \bu_S \bu_T 
\end{equation}
over all $\bu \in \RR^{n \choose \ell}$. To do this, we iterate the update equations
\begin{equation}\label{eq:update-u}
\bu_S \leftarrow \sum_{T \,:\, |S \symd T|=p} \bY_{S \symd T} \bu_T.
\end{equation}
We call $S$ and $T$ \emph{neighbors} if $|S \symd T|= p$. Intuitively, each neighbor $T$ of $S$ sends a message $m_{T \to S} \defeq \bY_{S \symd T} \bu_T$ to $S$, indicating $T$'s opinion about $\bu_S$. We update $\bu_S$ to be the sum of all incoming messages. 

Now note that the sum in~\eqref{eq:umu} is simply $\|u\|^{-2} \,\bu^\top \bM \bu$ where $\bM$ is the symmetric difference matrix, and~\eqref{eq:update-u} can be written as
\[\bu \leftarrow \bM \bu \, . \]
Thus, this natural message-passing scheme is precisely power iteration against $\bM$, and so we should take the leading eigenvector $\bv_{\textup{top}}$ of $\bM$ as our estimate of $(\bx_*^S)_{|S|=\ell}$ (up to a scaling factor). Finally, defining our voting matrix $\bV(\bv_{\textup{top}})$ and taking its leading eigenvector is a natural method for rounding $\bv_{\textup{top}}$ to a vector of the form $\bu^x$ where $u^x_S = \bx^S$, thus restoring the consistency constraints we ignored before.

For insight on why this ought to work, it is instructive to note that if we carry out this procedure on \emph{all} subsets $S \subseteq [n]$ then this works as intended, and no rounding is necessary: consider the $2^n \times 2^n$ matrix $\overline{\bM}_{S,T} = \bY_{S \symd T} \bone_{|S \symd T|=p}$. 
It is easy to verify that the eigenvectors of $\overline{\bM}$ are precisely the Fourier basis vectors on the hypercube, namely vectors of the form $\bu^\bx$ where $\bu^\bx_S=\bx^S$ and $\bx \in \{\pm 1\}^n$. Moreover, the eigenvalue associated to $\bu^\bx$ is
\[
\frac{1}{2^n} (\bu^\bx)^\top \overline{\bM} \bu^\bx 
= \frac{1}{2^n} \sum_{S, T \subseteq [n] \,:\, |S \symd T|=p} \bY_{S \symd T} \bx^S \bx^T 
= \sum_{|E|=p} \bY_E \bx^E \, .
\]
This is the expression in the log-likelihood in~\eqref{eq:posterior}. Thus, the leading eigenvector of $\overline{\bM}$ is exactly $\bu^\bx$ where $\bx$ is the maximum likelihood estimate of $\bx_*$.

This procedure succeeds all the way down to the information-theoretic threshold $\lambda \sim n^{(1-p)/2}$, but takes exponential time. Our contribution can be viewed as showing that even when we restrict to the submatrix $\bM$ of $\overline{\bM}$ supported on subsets of size $\ell$, the leading eigenvector still allows us to recover $\bx_*$ whenever the SNR is sufficiently large. Proving this requires us to perform Fourier analysis over a slice of the hypercube rather than the simpler setting of the entire hypercube, which we do by appealing to Johnson schemes and some results of~\cite{filmus}.

\subsection{Variational Inference and Kikuchi Free Energy} \label{sec:var-kik}

We now introduce the \emph{Kikuchi approximations to the free energy} (or simply the \emph{Kikuchi free energies}) of the above posterior~\eqref{eq:posterior} \cite{kikuchi1,kikuchi2}, the principle from which our algorithms are derived. For concreteness we restrict to the case of the Rademacher-spiked tensor model, but the Kikuchi free energies can be defined for general graphical models~\cite{GBP}.  

The posterior distribution in~\eqref{eq:posterior} is a Gibbs distribution $\PP(\bx \mid \bY) \propto e^{-\beta H(\bx)}$ with random Hamiltonian $H(\bx):=  - \lambda \sum_{i_1 < \cdots < i_p} \bY_{i_1 \cdots i_p} \bx_{i_1} \cdots \bx_{i_p}$, and inverse temperature $\beta = 1$. 
We let $Z_n(\beta;\bY) := \sum_{\bx \in \{\pm 1\}^n} e^{-\beta H(\bx)}$ be the partition function of the model, and denote by $F_n(\beta;\bY) : =  -\frac{1}{\beta}\log Z_n(\beta;\bY)$ its \emph{free energy}. 
It is a classical fact that the Gibbs distribution has the following variational characterization.
Fix a finite domain $\Omega$ (e.g., $\{\pm 1\}^n$), $\beta > 0$ and $H: \Omega \to \RR$. Consider the optimization problem
\begin{equation}\label{eq:variational}
\inf_{\mu} F(\mu),
\end{equation}
where the infimum is over probability distributions $\mu$ supported on $\Omega$, and define the \emph{free energy functional} $F$ of $\mu$ by
\begin{equation}\label{eq:F}
F(\mu) \defeq \Ex_{\bx \sim \mu}[H(\bx)] - \frac{1}{\beta} \bar{\mathcal{S}}(\mu),
\end{equation}
where $\bar{\mathcal{S}}(\mu)$ is the Shannon entropy of $\mu$, i.e., $\bar{\mathcal{S}}(\mu) = -\sum_{\bx \in \Omega} \mu(\bx) \log \mu(\bx)$. Then the unique optimizer of \eqref{eq:variational} is the Gibbs distribution $\mu^*(\bx) \propto \exp(-\beta H(\bx))$. 
If we specialize this statement to our setting, $\mu^* = \PP(\,\cdot\, | \bY)$ and $F_n(1;\bY) = F(\mu^*)$. We refer to~\cite{wainwright2008graphical} for more background.

In light of the above variational characterization, a natural algorithmic strategy to learn the posterior distribution is to minimize the free energy functional $F(\mu)$ over distributions $\mu$. However, this is \emph{a priori} intractable because (for a high-dimensional domain such as $\Omega = \{\pm 1\}^n$) an exponential number of parameters are required to represent $\mu$. The idea underlying the \emph{belief propagation} algorithm \cite{pearl-bp,GBP} is to work only with locally-consistent marginals, or \emph{beliefs}, instead of a complete distribution $\mu$. Standard belief propagation works with beliefs on singleton variables and on pairs of variables. The \emph{Bethe free energy} is a proxy for the free energy that only depends on these beliefs, and belief propagation is a certain procedure that iteratively updates the beliefs in order to locally minimize the Bethe free energy. The level-$r$ \emph{Kikuchi free energy} is a generalization of the \emph{Bethe free energy} that depends on $r$-wise beliefs and gives (in principle) increasingly better approximations of $F(\mu^*)$ as $r$ increases. Our algorithms are based on the principle of locally minimizing Kikuchi free energy, which we define next.

We now define the level-$r$ Kikuchi approximation to the free energy. We require $r \ge p$, i.e., the Kikuchi level needs to be at least as large as the interactions present in the data (although the $r < p$ case could be handled by defining a modified graphical model with auxiliary variables). The Bethe free energy is the case $r = 2$.

For $S \subseteq [n]$ with $0 < |S| \le r$, let $b_S: \{\pm 1\}^{S} \to \RR$ denote the \emph{belief} on $S$, which is a probability mass function over $\{\pm 1\}^{|S|}$ representing our belief about the joint distribution of $\bx_S \defeq \{\bx_i\}_{i \in S}$. Let $b = \{b_S :S \in {[n]\choose \le r}\}$ denote the set of beliefs on $s$-wise interactions for all $s \le r$.
Following \cite{GBP}, the Kikuchi free energy is a real-valued functional $\mathcal{K}$ of $b$ having the form $\mathcal{E} - \frac{1}{\beta}\mathcal{S}$ (in our case, $\beta = 1$). Here the first term is the `energy' term 
\[\mathcal{E}(b) = - \lambda \sum_{|S|=p} \bY_S \sum_{\bx_S \in \{\pm 1\}^S} b_S(\bx_S) \bx^S.\]
where, recall, $x^S \defeq \prod_{i \in S} x_i$. (This is a proxy for the term $\EE_{x \sim \mu} [H(x)]$ in \eqref{eq:F}.) The second term in $\mathcal{K}$ is the `entropy' term
\begin{equation}\label{eq:entrop}
\mathcal{S}(b) = \sum_{0 < |S| \le r} c_S \mathcal{S}_S(b), \qquad \mathcal{S}_S(b) = -\sum_{\bx_S \in \{\pm 1\}^S} b_S(\bx_S) \log b_S(\bx_S),
\end{equation}
where the \emph{overcounting numbers} are $c_S := \sum_{T \supseteq S,\, |T| \le r} (-1)^{|T \setminus S|}$. These are defined so that for any $S \subseteq [n]$ with $0 < |S| \le r$,
\begin{equation}\label{eq:counting}
\sum_{T \supseteq S,\, |T| \le r} c_T = 1,
\end{equation}
which corrects for overcounting the contribution of every variable $x_i$ to the entropy. The simplest case is when only beliefs on single marginals are considered: $r=1$. In this case the entropy term becomes $\mathcal{S}(b) = \sum_{i=1}^n - b(x_i)\log b(x_i)$, and the resulting approximation to the optimization problem \eqref{eq:variational} is known as the \emph{naive mean-field} approximation, which effectively optimizes over product measures $\mu(x) = \prod_{i=1}^n \mu_i(x_i)$. The case $r=2$ amounts to considering pair interactions, it considers distributions $\mu$ which can be approximated by an \emph{acyclic} Markov random field `at the level of the entropy': a Markov random field on a tree $T = (V=[n], E)$ (a graph without cycles) can be written in terms of its vertex and edge marginals as $\mu(x) = \prod_{i \in V} \mu(x_i) \prod_{(i,j) \in E} \mu(x_i,x_j)/[\mu(x_i)\mu(x_j)]$. The Shannon entropy of $\mu$ can be readily computed and is given by the expression $\bar{\mathcal{S}}(\mu) = \sum_{(i,j) \in E} \mathcal{S}_{\{i,j\}}(\mu) -\sum_{i \in V} (d_i-1)\mathcal{S}_{\{i\}}(\mu) = \mathcal{S}(\mu)$ where $d_i$ is the degree of vertex $i$ in $T$, and in writing $\mathcal{S}(\mu), \mathcal{S}_{\{i\}}(\mu)$ and $\mathcal{S}_{\{i,j\}}(\mu)$, we identify the probability distribution $\mu$ with its collection of marginals $\mu_S = \mu((x_i)_{i\in S})$, $S \subseteq [n]$. The Bethe approximation consists in ignoring the cycles in the graph---in our case the complete graph---and using an entropy of the previous form. The origin of the overcounting terms $c_S$ in this simple case is as follows: heuristically, in the pairwise terms $\sum_{(i,j) \in E} \mathcal{S}_{\{i,j\}}$ defining $\mathcal{S}(b)$ in Eq.~\eqref{eq:entrop}, each vertex $i$ `contributes' $d_i$ times to the entropy, so its marginal entropy must be subtracted $d_i-1$ times so that every vertex `contributes' once. A generalization of this approach to larger $r$ can be done via the language of hypergraphs and using hypertrees as a basis for the approximation. We do not attempt further elaboration here and refer to the excellent book~\cite[Section 4.2]{wainwright2008graphical} for a clear and full exposition of this formalism.

Notice that $\mathcal{E}$ and $\mathcal{S}$ each take the form of an ``expectation'' with respect to the beliefs $b_S$; these would be actual expectations were the beliefs the marginals of an actual probability distribution. This situation is to be compared with the notion of a \emph{pseudo-expectation}, which plays a central role in the theory underlying the sum-of-squares algorithm.

Our algorithms are based on the \emph{Kikuchi Hessian}, a generalization of the Bethe Hessian matrix that was introduced in the setting of community detection \cite{bethe-hessian}. The Bethe Hessian is the Hessian of the Bethe free energy with respect to the moments of the beliefs, evaluated at belief propagation's so-called ``uninformative fixed point.'' The bottom eigenvector of the Bethe Hessian is a natural estimator for the planted signal because it represents the best direction for local improvement of the Bethe free energy, starting from belief propagation's uninformative starting point. We generalize this method and compute the analogous Kikuchi Hessian matrix.
The full derivation is given in Appendix~\ref{app:hessian}. The order-$\ell$ symmetric difference matrix \eqref{symmetric_diff} (approximately) appears as a submatrix of the level-$r$ Kikuchi Hessian whenever $r \ge \ell + p/2$.

\section{Analysis of Symmetric Difference and Voting Matrices}
\label{sec:symm_diff}

We adopt the notation $\bx^{S} \defeq \prod_{i\in S} \bx_i$ for $\bx \in \{\pm 1\}^n$ and $S \subseteq [n]$. Recall the matrix $\bM$ indexed by sets $S \subseteq [n]$ of size $\ell$, having entries 
\begin{equation}
\bM_{S,T} = \bY_{S \symd T} \bone_{|S \symd T| = p}
\quad \text{where} \quad  
\bY_{S \symd T} = \lambda x_*^{S \symd T} + \bG_{S \symd T} \, . 
\end{equation}   
First, observe that we can restrict our attention to the case where the spike is the all-ones vector $x_* = \One$ without loss of generality. To see this, conjugate $\bM$ by a diagonal matrix $\bD$ with diagonal entries $\bD_{S,S} = x_*^S$ and obtain $(\bM')_{S,T} = (\bD^{-1}\bM \bD)_{S,T} = \bY'_{S \symd T} \bone_{|S \symd T| = p}$ where $\bY'_{S \symd T} = x_*^S x_*^T \bY_{S \symd T} = x_*^{S \symd T} \bY_{S \symd T} = \lambda + G'_{S \symd T}$ where $G'_{S \symd T} = x_*^{S}x_*^{T}G_{S \symd T}$. By symmetry of the Gaussian distribution, $(G'_{E})_{|E|=p}$ are i.i.d.\ $\sN(0,1)$ random variables.  
Therefore, the two matrices have the same spectrum and the eigenvectors of $\bM$ can be obtained from those of $\bM'$ by pre-multiplying by $\bD$. Similarly, the case $x_* = \One$ is sufficient for the analysis of the voting matrices. So from now on we write 
\begin{equation}\label{signal_noise}
\bM = \lambda \bX + \bZ,
\end{equation}
where $\bX_{S,T} = \bone_{|S \symd T|=p}$ and $\bZ_{S,T} = G_{S \symd T}\bone_{|S \symd T|=p}$, where $(G_{E})_{|E|=p}$ is a collection of i.i.d.\ $\sN(0,1)$ r.v.'s.

\subsection{Structure of \emph{X}}

The matrix $\bX$ is the adjacency matrix of a regular graph $J_{n,\ell,p}$ on ${n \choose \ell}$ vertices, where vertices are represented by sets, and two sets $S$ and $T$ are connected by an edge if $|S\symd T|=p$, or equivalently $|S \cap T| = \ell - p/2$. 
This matrix belongs to the Bose-Mesner algebra of the $(n,\ell)$-Johnson association scheme (see for instance~\cite{schrijver,godsil2010association}). This is the algebra of ${n \choose \ell} \times {n \choose \ell}$ real- or complex-valued symmetric matrices where the entry $\bX_{S,T}$ depends only on the size of the intersection $|S \cap T|$. In addition to this set of matrices being an algebra, it is a \emph{commutative} algebra, which means that all such matrices are simultaneously diagonalizable and share the same eigenvectors. 

Filmus~\cite{filmus} provides a common eigenbasis for this algebra: for $0 \le m \le \ell$, let $\varphi = (a_1,b_1,\ldots,a_m,b_m)$ be a sequence of $2m$ distinct elements of $[n]$. Let  $|\varphi|=2m$ denote its total length. Now define a vector $u^{\varphi} \in \RR^{{n \choose \ell}}$ having coordinates    
\[
u^\varphi_S = \prod_{i=1}^m (\bone_{a_i \in S} - \bone_{b_i \in S}) \, , \quad |S| = \ell \, .
\]
In the case $m=0$, $\varphi$ is the empty sequence $\varnothing$ and we have $u^\varnothing = \One$ (the all-ones vector).

\begin{proposition}\label{eigenspaces}
Each $u^\varphi$ is an eigenvector of $\bX$. Furthermore, the linear space $\mathcal{Y}_m :=\text{span}\{u^{\varphi} : |\varphi|=2m\}$ for $0 \le m \le \ell$ is an eigenspace of $\bX$ (i.e., all vectors $u^\varphi$ with sequences $\varphi$ of length of $2m$ have the same eigenvalue $\mu_m$). Lastly $\RR^{{n \choose \ell}}= \bigoplus_{m=0}^{\ell}\mathcal{Y}_m$, and $\dim \mathcal{Y}_m = {n \choose m} - {n \choose m-1}$. (By convention, ${n \choose -1}=0$.)
\end{proposition}
\begin{proof}
The first two statements are the content of Lemma 4.3 in~\cite{filmus}. The dimension of $\mathcal{Y}_m$ is given in Lemma 2.1 in~\cite{filmus}.
\end{proof}
We note that $(u^\varphi)_{|\varphi| = 2m}$ are not linearly independent; an orthogonal basis, called the Young basis, consisting of linear combinations of the $u^\varphi$'s is given explicitly in~\cite{filmus}.

We see from the above proposition that $\bX$ has $\ell+1$ distinct eigenvalues $\mu_{0} \ge \mu_{1}\ge \cdots\ge \mu_{\ell}$, each one corresponding to the eigenspace $\mathcal{Y}_m$. The first eigenvalue is the degree of the graph $J_{n,\ell,p}$:  
\begin{equation}\label{eq:mu0}
\mu_0 = d_\ell = {\ell \choose p/2}{n-\ell \choose p/2} \, .
\end{equation}
We provide an explicit formula for all the remaining eigenvalues: 
\begin{lemma}
\label{lem:eberlein}
The eigenvalues of $\bX$ are as follows:
\begin{equation}
\label{eq:eig-mu}
\mu_m = \sum_{s=0}^{\min(m,p/2)} (-1)^s {m \choose s} {\ell-m \choose p/2-s} {n-\ell-m \choose p/2-s} \, , 
\quad 0 \le m \le \ell \, . 
\end{equation}
\end{lemma}

\begin{proof}
These are the so-called Eberlein polynomials, which are polynomials in $m$ of degree $p$ (see, e.g.,~\cite{schrijver}). We refer to~\cite{burcroffjohnson} for formulae in more general contexts, but we give a proof here for completeness. Let $A=\{a_1,\ldots,a_m\}$ and $B=\{b_1,\ldots,b_m\}$. Note that $u^\varphi_S$ is nonzero if and only if $|S \cap \{a_i,b_i\}|=1$ for each $1 \le i \le m$. By symmetry, we can assume that $A \subseteq S$ and $S \cap B=\varnothing$. Then,
\[ \mu_n = \sum_{T \,:\, |T| = \ell, \, |S \symd T|=p, \, |T \cap \{a_i,b_i\}|=1 \; \forall i} (-1)^{|T \cap B|}. \]
Letting $s = |T \cap B|$, we will count the number of sets $T$ appearing in the sum for a fixed value of $s$. First, there are $\binom{m}{s}$ ways to choose which elements of $A \sqcup B$ belong to $T$. Next, to achieve $|S \symd T|=p$, $T$ must include exactly $p/2-s$ elements of $\overline{S \cup B}$, and there are $\binom{n-\ell-m}{p/2-s}$ way to choose these. Finally, the remaining $\ell-m-p/2+s$ elements of $T$ must come from $S \setminus A$, and there are $\binom{\ell-m}{\ell-m-p/2+s} = \binom{\ell-m}{p/2-s}$ ways to choose these. Also note that we must have $s \le m$ and $s \le p/2$, giving the bounds for the sum.
\end{proof}

As the following lemma shows, the succeeding eigenvalues decay rapidly with $m$.

\begin{lemma}\label{lem:eig-gap} 
Let $3 \le p \le \sqrt{n}$ and let $\ell < n/p^2$.  
For all $0 \le m \le \ell-p/2$ it holds that
\[\frac{|\mu_m|}{\mu_0} \le \left( 1-\frac{m}{\ell} \right)^{p/2} \, .\]
\end{lemma}
\begin{proof}
The terms in~\eqref{eq:eig-mu} have alternating signs. We will show that they decrease in absolute value beyond the first nonzero term, so that this gives a bound on $\mu_m$. Since $m \le \ell-p/2$, the $s=0$ term is positive. The $(s+1)$st term divided by the $s$th term is, in absolute value, 
\begin{align*}
\frac{{m \choose s+1} {\ell-m \choose p/2-s-1} {n-\ell-m \choose p/2-s-1}}{{m \choose s} {\ell-m \choose p/2-s} {n-\ell-m \choose p/2-s}}
&= \left( \frac{m-s}{s+1} \right)
\left( \frac{p/2-s}{\ell-m-p/2+s+1} \right) 
\left( \frac{p/2-s}{n-\ell-m-p/2+s+1} \right) \\ 
&\le \frac{m (p/2)(p/2)}
{(\ell-m-p/2+1) (n-\ell-m-p/2+1)} \\
&\le \frac{(\ell-p/2)(p/2)(p/2)}
{n-2\ell+1}  \quad \text{since $m \le \ell-p/2$} \\
 &\le \frac{\ell p^2/4}
{n-2\ell} < 1 \, .
\end{align*}
We used the assumptions $3 \le p \le \sqrt{n}$ and $\ell p^2 \le n$ to obtain the last bound.
It follows that the $s=0$ term is an upper bound, 
\[
\mu_m \le {\ell-m \choose p/2} {n-\ell-m \choose p/2} \, , 
\]
and so 
\[
\frac{\mu_m}{\mu_0} 
\le {\ell-m \choose p/2} \!\left\slash {\ell \choose p/2} \right. 
\le \left( \frac{\ell-m}{\ell} \right)^{\!p/2} \, .
\]
\end{proof}

\subsection{Proof of Strong Detection}

Here we prove our strong detection result, Theorem~\ref{detection_result}. The proof does not exploit the full details of the structure exhibited above. Instead, the proof is a straightforward application of the matrix Chernoff bound for Gaussian series~\cite{oliveira} (see also Theorem~4.1 of~\cite{tropp}):
\begin{theorem}\label{thm:matrix-conc}
Let $\{\bA_i\}$ be a finite sequence of fixed symmetric $d \times d$ matrices, and let $\{\xi_i\}$ be independent $\sN(0,1)$ random variables. Let ${\bm \Sigma} = \sum_i \xi_i \bA_i$. Then, for all $t \ge 0$,
\[\PP\left( \|{\bm \Sigma}\|_{\textup{op}}\ge t\right) \le 2d e^{-t^2/2\sigma^2} \quad\text{where } \sigma^2 =  \left\| \EE[{\bm \Sigma}^2]\right\|_{\textup{op}}.\]
Furthermore, the same bound holds if $\{\xi_i\}$ is a sequence of independent Rademacher random variables, i.e., $\pm 1$-valued with equal probabilities.  
\end{theorem}
  
Let us first write $\bM$ in the form of a Gaussian series. For a set $E \in {[n]\choose p}$, define the ${n \choose \ell} \times {n \choose \ell}$ matrix $\bA_{E}$ as
\[(\bA_{E})_{S,T} = \bone_{S \symd T = E} \, . \]   
It is immediate that for $\lambda=0$, $\bM = \bZ = \sum_{|E|=p} g_E \bA_E$ where $(g_E)_{|E|=p}$ is a collection of i.i.d.\ $\sN(0,1)$ random variables. The second moment of this random matrix is
\[\EE[\bM^2]= \sum_{E\,:\,|E|=p} \bA_E^2 = d_\ell \bI \, , \]
since $\bA_E^2$ is the diagonal matrix with $(\bA_E)_{S,S}=\bone_{|S \symd E|=p}$, and summing over all $E$ gives $d_\ell$ on the diagonal. 
The operator norm of the second moment is then $d_\ell$. It follows that for all $t \ge 0$,
\begin{equation}\label{eq:Z-bound}
\PP_{0}\big(\lambda_{\max}(\bM) \ge t\big) \le 2{n \choose \ell} e^{-t^2/2d_\ell}\le 2n^\ell e^{-t^2/2d_\ell}.
\end{equation}
Now letting $t= \frac{\lambda d_\ell}{2}$ yields the first statement of the theorem. 

As for the second statement, we have $\lambda_{\max}(\bM) \ge \|\bX\|_{\textup{op}} - \|\bZ\|_{\textup{op}} = \lambda d_{\ell} - \|\bZ\|_{\textup{op}}$ where $\bZ$ is defined in~\eqref{signal_noise}. Applying the same bound we have
\[\PP_{\lambda}\Big(\lambda_{\max}(\bM) \le \frac{\lambda d_\ell}{2}\Big) \le \PP\Big(\|\bZ\|_{\textup{op}} \ge  \frac{\lambda d_\ell}{2}\Big)
 \le 2n^{\ell}e^{- \lambda^2 d_\ell / 8}.\]

\subsection{Proof of Strong Recovery}
Here we prove our strong recovery result, Theorem~\ref{recovery_result}.  
Let $v_0 = \bv_{\textup{top}}(\bM)$ be a unit-norm leading eigenvector of $\bM$. For a fixed $m \in [\ell]$ to be determined later on, we write the orthogonal decomposition $v_0 = v^{(m)}+v^{\perp}$, where $v^{(m)} \in \bigoplus_{s \le m} \mathcal{Y}_s$, and $v^{\perp}$ in the orthogonal complement.
The goal is to first show that if $m$ is proportional to $\ell$ then $v^\perp$ has small Euclidean norm, so that $v_0$ and $v^{(m)}$ have high inner product. The second step of the argument is to approximate the  voting matrix $\bV(v_0)$ by $\bV(v^{(m)})$, and then use Fourier-analytic tools to reason about the latter. 

Let us start with the first step. 
\begin{lemma}\label{orthogonal_noise}
With $\bZ$ defined as in~\eqref{signal_noise}, we have \[\big\|v^{\perp}\big\|^2 \le 2\frac{\|\bZ\|_{\textup{op}}}{\lambda(\mu_0 - \mu_{m+1})}.\]
 \end{lemma}
 \begin{proof}
Let us absorb the factor $\lambda$ in the definition of the matrix $\bX$. Let $\{u_0,\cdots,u_d\}$ be a set of eigenvectors of $\bX$ which also form an orthogonal basis for $\bigoplus_{s \le m} \mathcal{Y}_s$, with $u_0$ being the top eigenvector of $\bX$ ($u_0$ is the normalized all-ones vector).  
We start with the inequality $u_0^\top \bM u_0 \le v_0^\top \bM v_0$.
The left-hand side of the inequality is $\mu_0 + u_0^\top \bZ u_0$.
The right-hand side is $v_0^\top \bX v_0 + v_0^\top \bZ v_0$. Moreover $v_0^\top \bX v_0 = v_0^{\top}\bX v^{(m)} +  v_0^{\top}\bX v^{\perp}$. Since $v^{(m)} \in \bigoplus_{s \le m}\mathcal{Y}_s$, by Proposition~\ref{eigenspaces}, $\bX v^{(m)} $ belongs to the space as well, and therefore $v_0^{\top}\bX v^{(m)} = v^{(m)\top}\bX v^{(m)}$. Similarly $v_0^{\top}\bX v^{\perp} = (v^{\perp})^{\top}\bX v^{\perp}$, so $v_0^\top \bX v_0 = v^{(m)\top}\bX v^{(m)} + (v^{\perp})^{\top}\bX v^{\perp}$.
Therefore the inequality becomes 
\[\mu_0 + u_0^\top \bZ u_0 \le v^{(m)\top}\bX v^{(m)}+ (v^{\perp})^{\top}\bX v^{\perp} + v_0^\top \bZ v_0.\]
Since $v^\perp$ is orthogonal to the top $m$ eigenspaces of $\bX$ we have $(v^{\perp})^{\top}\bX v^{\perp} \le \mu_{m+1} \big\|v^{\perp}\big\|_2^2$. Moreover, $v^{(m)\top}\bX v^{(m)} \le \mu_0 \big\|v^{(m)}\big\|^2$, hence
\[\mu_0 + u_0^\top \bZ u_0 \le \mu_0 \big\|v^{(m)}\big\|^2+ \mu_{m+1} \big\|v^{\perp}\big\|^2 + v_0^\top \bZ v_0.\]
By rearranging and applying the triangle inequality we get,
\[(\mu_0 - \mu_{m+1}) \big\|v^{\perp}\big\|^2   \le |v_0^\top \bZ v_0| + |u_0^\top \bZ u_0| \le 2 \|\bZ\|_{\textup{op}}.\qedhere\]
\end{proof}

Combining this fact with Lemma~\ref{lem:eig-gap} and recalling that $\mu_0=d_\ell$, we obtain the following result: 
\begin{lemma}\label{lem:v_perpendicular}
For any $\epsilon>0$ and $\delta\in (0,1)$, if $\lambda \ge \epsilon^{-1} \sqrt{2\log(n^\ell/\delta) / d_\ell}$, and $m+1 \le \ell - p/2$ then
\[\big\|v^{\perp}\big\|^2 \le \epsilon\frac{\ell}{m+1}\]
with probability at least $1-\delta$.
\end{lemma}
\begin{proof}
Using Lemma~\ref{lem:eig-gap} implies $\frac{1}{\mu_0 - \mu_{m+1}} \le \frac{1}{\mu_0} \frac{1}{1-(1-\frac{m+1}{\ell})^{p/2}} \le \frac{1}{\mu_0} \cdot  \frac{\ell}{m+1}$. To obtain the last bound we use the inequalities  $\frac{1}{1-(1-x)^\alpha} \le \frac{1}{1-e^{-\alpha x}} < \frac{1}{x}$, the last one, obtained by simple calculus, is valid for all $\alpha>1$ and $0 < x < x_0(\alpha) = \alpha^{-1}\log \alpha$. We let $x = (m+1)/\ell$ and $\alpha = p/2>1$. 
Therefore, Lemma~\ref{orthogonal_noise} implies
\[\big\|v^{\perp}\big\|^2 \le \frac{\|\bZ\|_{\textup{op}}}{\lambda d_{\ell}}  \frac{\ell}{m+1}.\]
The operator norm of the noise can be bounded by a matrix Chernoff bound~\cite{oliveira,tropp}, similarly to our argument in the proof of detection: for all $t \ge 0$,
 \[\PP(\|\bZ\|_{\textup{op}} \ge t) \le n^\ell e^{-t^2/2d_\ell}.\]  
Therefore, letting $\lambda \ge \epsilon^{-1}\sqrt{2\log(n^\ell/\delta)/d_{\ell}}$ we obtain the desired result. 
 \end{proof}
 
\subsubsection{Analysis of the Voting Matrix} 
Recall that the voting matrix $\bV(v)$ of a vector $v \in \RR^{{n \choose \ell}}$ has zeros on the diagonal, and off-diagonal entries  
\begin{align*}
\bV_{ij}(v) &= \frac{1}{2}\sum_{S,T \in  {[n] \choose \ell}} v_S v_T \bone_{S \symd T = \{i,j\}} \\
&= \sum_{S\in  {[n] \choose \ell}} v_S v_{S \symd \{i,j\}} \bone_{i \in S, j\notin S} \, , \quad 1 \le i \neq j \le n \, .
\end{align*}
It will be more convenient in our analysis to work with $\bV(v^{(m)})$ instead of $\bV(v_0)$. To this end we produce the following approximation result:
\begin{lemma}\label{lem:voting-noise}
Let $u,e \in \RR^{n \choose \ell}$ and $v = u + e$. Then
\[\|\bV(v) - \bV(u)\|_F^2 \le 3 \ell^2 \|e\|^2 (2 \|u\|^2 + \|e\|^2).\]
In particular,
\[\big\|\bV(v_0) - \bV(v^{(m)})\big\|_F^2 \le 9 \ell^2 \big\|v^{\perp}\big\|^2.\]
\end{lemma}
\begin{proof}
Let us introduce the shorthand notation $\langle u,v \rangle_{ij} \defeq \sum_{S\in  {[n] \choose \ell}} u_S v_{S \symd \{i,j\}} \bone_{i \in S, j\notin S}$. We have
\begin{align}
\big\|\bV(v) - \bV(u)\big\|_F^2 &= \sum_{i,j} (\bV_{ij}(u+e) - \bV_{ij}(u))^2 \nonumber\\
&= \sum_{i,j} \left(\langle u+e,u+e \rangle_{ij} - \langle u,u \rangle_{ij}\right)^2 \nonumber\\
&= \sum_{i,j} \left(\langle u,e \rangle_{ij} + \langle e,u \rangle_{ij} + \langle e,e \rangle_{ij}\right)^2 \nonumber\\
&\le 3 \sum_{i,j} \left(\langle u,e \rangle_{ij}^2 + \langle e,u \rangle_{ij}^2 + \langle e,e \rangle_{ij}^2\right),
\label{eq:VF}
\end{align}
where the last step uses the bound $(a+b+c)^2 \le 3(a^2+b^2+c^2)$.
Now we expand
\begin{align}
\sum_{i,j} \langle u,e \rangle_{ij}^2 &= \sum_{i,j} \left(\sum_{S\,:\,|S|=\ell,\,i \in S,\,j \notin S} u_S e_{S \symd \{i,j\}}\right)^2 \nonumber\\
&\le \sum_{i,j} \left(\sum_{S\,:\,|S|=\ell,\,i \in S,\,j \notin S} u_S^2\right)\left(\sum_{S\,:\,|S|=\ell,\,i \in S,\,j \notin S} e^2_{S \symd \{i,j\}}\right) \quad\text{(Cauchy--Schwarz)}\nonumber\\
&\le \sum_{i,j} \left(\sum_{S\,:\,|S|=\ell,\,i \in S} u_S^2\right)\left(\sum_{T\,:\,|T|=\ell,\,j \in T} e^2_T\right) \nonumber\\
&= \left(\sum_i \sum_{S\,:\,|S|=\ell,\,i \in S} u_S^2\right)\left(\sum_j \sum_{T\,:\,|T|=\ell,\,j \in T} e^2_T\right) \nonumber\\
&= \left(\ell \sum_{|S|=\ell} u_S^2\right)\left(\ell \sum_{|T|=\ell} e_T^2\right) \nonumber\\
&= \ell^2 \|u\|^2 \|e\|^2.
\label{eq:CS}
\end{align}
Plugging this back into \eqref{eq:VF} yields the desired result.
To obtain $\|\bV(v_0) - \bV(v^{(m)})\|_F^2 \le 9 \ell^2 \|v^{\perp}\|^2$ we just bound $2\|v^{(m)}\|^2 + \|v^\perp\|^2$ by $3$.
\end{proof}

 Let us also state the following lemma, which will be need later on:
\begin{lemma}\label{lem:voting-norm}
For $u \in \RR^{n \choose \ell}$, $\|\bV(u)\|_F^2 \le \ell^2 \|u\|^4$.
\end{lemma}
\begin{proof}
Note that $\|\bV(u)\|_F^2 = \sum_{i,j} \bV_{ij}(u)^2 = \sum_{i,j} \langle u,u \rangle_{ij}^2$ and so the desired result follows immediately from \eqref{eq:CS}.
\end{proof}

Next, in the main technical part of the proof, we show that $\bV(v^{(m)})$ is close to a multiple of the all-ones matrix in Frobenius norm:
\begin{proposition}\label{lem:approx_rank_one}
Let $\hat\One = \One/\sqrt{n}$ and $\alpha = \ell \|v^{(m)}\|^2$. Then
\[\big\|\bV(v^{(m)}) - \alpha \hat\One\hat\One^\top\big\|_F^2 \le 2 \left(\frac{m}{\ell} + \frac{\ell}{n}\right) \alpha^2.\]
\end{proposition}

Before proving the above proposition, let us put the results together and prove our recovery result.

\subsubsection{Proof of Theorem~\ref{recovery_result}} 
For $\epsilon, \delta>0$, assume $\lambda \ge \epsilon^{-1} \sqrt{2\log(n^\ell/\delta) / d_\ell}$ and $m+1 \le \ell - p/2$. By Lemma~\ref{lem:v_perpendicular} and Lemma~\ref{lem:voting-noise} we have
\[\big\|\bV(v_0) - \bV(v^{(m)})\big\|_F^2 \le 9 \ell^2 \frac{\ell}{m+1}\epsilon \]
with probability at least $1-\delta$. Moreover, by Proposition~\ref{lem:approx_rank_one}, we have
\[\big\|\bV(v^{(m)}) - \alpha \hat\One\hat\One^\top\big\|_F^2 \le 2 \eta \alpha^2,\]
with $\alpha = \ell \|v^{(m)}\|^2 \le \ell$ and $\eta = \frac{m}{\ell} + \frac{\ell}{n}$.
Therefore, by a triangle inequality we have
\begin{align*}
\big\|\bV(v_0) - \alpha \hat\One\hat\One^\top\big\|_F &\le  3\ell \sqrt{\frac{\ell\epsilon}{m+1}}+ \sqrt{2 \eta} \, \alpha \\
&\le 3\ell \Big( \sqrt{\frac{\ell\epsilon}{m+1}} + \sqrt{\frac{m}{\ell} + \frac{\ell}{n}}\Big),
\end{align*}
with probability at least $1-\delta$. Now let us choose a value of $m$ that achieves a good tradeoff of the above two error terms: we take $m = \lfloor\ell\sqrt{\epsilon}\rfloor$ if this choice satisfies the condition $m+1 \le \ell - p/2$, i.e, if $\ell \ge (p/2+1)/(1-\sqrt{\epsilon})$. Otherwise, we take $m=0$. 
Use the inequality $\sqrt{a+b} \le \sqrt{a}+\sqrt{b}$ for positive $a,b$, we obtain
\begin{align}
\big\|\bV(v_0) - \alpha \hat\One\hat\One^\top\big\|_F &\le 3\ell \max\Big\{\Big(2\epsilon^{1/4} + \sqrt{\frac{\ell}{n}}\Big), c(p)\Big(\sqrt{\epsilon} + \sqrt{\frac{1}{n}}\Big)\Big\} \label{frobenius_error0},\\
&\le 6\ell\Big(\epsilon^{1/4} + \sqrt{\frac{\ell}{n}}\Big)\,,\label{frobenius_error}
\end{align}
where $c(p)$ is a constant depending only on $p$ when $\epsilon$ is bounded away from 1. (The first term in the above upper bound corresponds to the case $m = \lfloor \ell \sqrt{\epsilon}\rfloor, \ell \ge (p/2+1)/(1-\sqrt{\epsilon})$, and the second term to $m = 0, \ell \le (p/2+1)/(1-\sqrt{\epsilon})$.)
For small $\epsilon$ and large $n$ the first term in~\eqref{frobenius_error0} dominates, leading to the bound~\eqref{frobenius_error}. 

Now let $\widehat{x}$ be a leading eigenvector of $\bV(v_0)$, and let $\bR = \bV(v_0) - \alpha \hat\One\hat\One^\top$. Since $\widehat{x}^\top \bV(v_0)\widehat{x} \ge \hat{\One}^\top \bV(v_0)\hat{\One}$, we have
\[\alpha \langle \widehat{x}, \hat{\One}\rangle^2 + \widehat{x}^\top \bR \widehat{x} \ge \alpha + \hat{\One}^\top \bR\hat{\One}.\] 
Therefore
\[\langle \widehat{x}, \hat{\One}\rangle^2 \ge 1 - 2 \frac{\|\bR\|_{\textup{op}}}{\alpha}.\]
Since $\alpha = \ell(1-\|v^\perp\|^2)$, and $\|v^{\perp}\|^2 \le \epsilon \frac{\ell}{m+1} \le \max\{\sqrt{\epsilon}, c(p)\epsilon\} \le \sqrt{\epsilon}$ for small $\epsilon$, 
the bound~\eqref{frobenius_error} (together with $\|\bR\|_{\textup{op}}\le \|\bR\|_{F}$) implies 
\begin{equation}\label{correlation_bound_final}
\langle \widehat{x}, \hat{\One}\rangle^2 \ge 1 - 12\frac{\epsilon^{1/4} + \sqrt{\ell/n}}{1-\sqrt{\epsilon}}.
\end{equation}
To conclude the proof of our theorem, we let $\ell \le n\sqrt{\epsilon} $, $\epsilon < 1/16$ and then replace $\epsilon$ by $\epsilon^{4}$: we obtain $\langle \widehat{x}, \hat{\One}\rangle^2 \ge 1 - 48 \epsilon$ with probability at least $1-\delta$ if $\lambda \ge  \epsilon^{-4} \sqrt{2\log(n^\ell/\delta) / d_\ell}$.

\subsubsection{Proof of Proposition~\ref{lem:approx_rank_one}}

Let $\alpha = \ell \|v^{(m)}\|^2$. By Lemma~\ref{lem:voting-norm} we have $\big\|\bV(v^{(m)})\big\|_F^2\le \alpha^2$. Therefore
\begin{align}
\big\|\bV(v^{(m)}) - \alpha \hat\One\hat\One^\top\big\|_F^2 &= \big\|\bV(v^{(m)})\big\|_F^2 - \frac{2}{n} \alpha \sum_{i,j=1}^n \bV_{ij}(v^{(m)})+ \alpha^2 \nonumber\\
&\le 2\alpha^2 - \frac{2}{n} \alpha \sum_{i,j=1}^n \bV_{ij}(v^{(m)}).\label{fro_bound}
\end{align}
Now we need a lower bound on $\sum_{i,j=1}^n \bV_{ij}(v^{(m)})$. This will crucially rely on the fact that $v^{(m)}$ lies in the span of the top $m$ eigenspaces of $\bX$:

\begin{lemma}\label{lem:voting-fourier}
For a fixed $m \le \ell$,  let $v \in \bigoplus_{s=0}^{m}\mathcal{Y}_s$. Then
\[\sum_{i,j=1}^n \bV_{ij}(v) \ge ((\ell - m)n - \ell^2) \,\|v\|^2.\]
\end{lemma}

We substitute the result of the above lemma in~\eqref{fro_bound} to obtain
\begin{align*}
\big\|\bV(v^{(m)}) - \alpha \hat\One\hat\One^\top\big\|_F^2  &\le 2\alpha^2\Big(1 - \frac{\ell - m}{\ell}+\frac{\ell}{n}\Big)\\
&=2\alpha^2\Big(\frac{m}{\ell}+\frac{\ell}{n}\Big),
\end{align*}
as desired. 

\subsubsection{A Poincar\'e Inequality on a Slice of the Hypercube}
To prove Lemma~\ref{lem:voting-fourier}, we need some results on Fourier analysis on the slice of the hypercube ${[n]\choose \ell}$. Following~\cite{filmus}, we define the following. 
First, given a function $f : {[n]\choose \ell} \mapsto \RR$, we define its expectation as its average value over all sets of size $\ell$, and write
\[\Ex_{|S|=\ell}[f(S)] := \frac{1}{{n\choose \ell}}\sum_{|S|=\ell} f(S).\]  
We also define its variance as
\[\VV[f] := \Ex_{|S|=\ell} [f(S)^2] - \Ex_{|S|=\ell} [f(S)]^2.\]
Moreover, we identify a vector $u \in \RR^{{n \choose \ell}}$ with a function on sets of size $\ell$ in the obvious way: $f(S) = u_S$.
\begin{definition}
For $u \in \RR^{n \choose \ell}$ and $i,j \in [n]$, let $u^{(ij)}$ denote the vector having coordinates
\[u^{(ij)}_S = \left\{\begin{array}{ll} u_{S \symd \{i,j\}} & \text{if } |S \cap \{i,j\}| = 1 \\ u_S & \text{otherwise.}\end{array}\right.\]
(The operation $u \mapsto u^{(ij)}$ exchanges the roles of $i$ and $j$ whenever possible.) 
\end{definition}

\begin{definition}
For $u \in \RR^{n \choose \ell}$ and $i,j \in [n]$, define the \emph{influence} of the pair $(i,j)$ as
\[\Inf_{ij}[u] := \frac{1}{2} \Ex_{|S|=\ell} \Big[\big(u^{(ij)}_S - u_S\big)^2\Big],\]
and the \emph{total influence} as
\[\Inf[u] := \frac{1}{n} \sum_{i < j} \Inf_{ij}[u].\]
\end{definition}

With this notation, our main tool is a version of Poincar\'e's inequality on ${[n]\choose \ell}$:
\begin{lemma}[Lemma~5.6 in~\cite{filmus}]
For $v \in  \bigoplus_{s=0}^{m}\mathcal{Y}_s$ we have
\[\Inf[v] \le m \VV[v].\]
\end{lemma} 

\vspace{.25cm}
\begin{proof}[Proof of Lemma~\ref{lem:voting-fourier}]
We have $m \Ex_{|S|=\ell} [v_S^2] \ge m \VV[v] \ge \Inf[v]=\frac{1}{n} \sum_{i < j} \Inf_{ij}[v]$, and
\begin{align*}
2\,\Inf_{ij}[v] &=\Ex_{|S|=\ell} \big[(v_S^{(ij)} - v_S)^2\big] \\
&= \Ex_{|S|=\ell} \big[\bone_{|S \cap \{i,j\}|=1} (v_S^{(ij)} - v_S)^2\big]\\
&=\Ex_{|S|=\ell} \Big[\bone_{|S \cap \{i,j\}|=1} \Big((v_S^{(ij)})^2 - 2 v_S^{(ij)} v_S + v_S^2\Big)\Big].
\end{align*}
Since $\Ex_{|S|=\ell} [\bone_{|S \cap \{i,j\}|=1} (v_S^{(ij)})^2] = \Ex_{|S|=\ell} [\bone_{|S \cap \{i,j\}|=1} v_S^2]$, the above is equal to
\[2\Ex_{|S|=\ell} \Big[\bone_{|S \cap \{i,j\}|=1} \big(v_S^2 - v_S^{(ij)} v_S\big)\Big].\]
Now,
\begin{align*}
\frac{1}{n} \sum_{i < j} \Inf_{ij}[v]
&= \frac{1}{n} \sum_{i < j} {n \choose \ell}^{-1} \sum_{S \,:\, |S|=\ell,\,|S \cap \{i,j\}|=1} \big( v_S^2 -  v_S^{(ij)} v_S\big) \\
&= \frac{1}{n} {n \choose \ell}^{-1} \left[ \ell(n-\ell)\sum_{|S|=\ell} v_S^2 - \sum_{i,j \,:\, i \ne j} \;\sum_{S\,:\,|S|=\ell,\,i \in S,\,j \notin S} v_S v_{S \symd \{i,j\}}\right]\\
&= \frac{1}{n} {n \choose \ell}^{-1} \left[ \ell(n-\ell)\sum_{|S|=\ell} v_S^2 - \sum_{i,j} \bV_{ij}(v)\right] \\
&= \frac{1}{n} {n \choose \ell}^{-1} \left[ \ell(n-\ell)\|v\|^2 - \sum_{i,j} \bV_{ij}(v)\right].
\end{align*}
Therefore \[m \|v\|^2 \ge \frac{1}{n} \Big(\ell(n-\ell)\|v\|^2 - \sum_{i,j} \bV_{ij}(v)\Big),\]
and so
\[\sum_{i,j} \bV_{ij}(v) \ge (\ell(n-\ell) - mn)\, \|v\|^2 = ((\ell - m)n - \ell^2)\, \|v\|^2.\]
as desired.
\end{proof}

\section{Extensions}

\subsection{Refuting Random $k$-XOR Formulas for $k$ Even}
\label{sec:xor}

Our symmetric difference matrices can be used to give a simple algorithm and proof for a related problem: strongly refuting random $k$-XOR formulas (see~\cite{refute-csp,strongly-refuting} and references therein). This is essentially a variant of the spiked tensor problem with sparse Rademacher observations instead of Gaussian ones. It is known \cite{strongly-refuting} that this problem exhibits a smooth tradeoff between subexponential runtime and the number of constraints required, but the proof of \cite{strongly-refuting} involves intensive moment calculations. When $k$ is even, we will give a simple algorithm and a simple proof using the matrix Chernoff bound that achieves the same tradeoff. SOS lower bounds suggest that this tradeoff is optimal~\cite{grig-sos,sch-sos}. 

After the initial appearance of our work, the follow-up work~\cite{smoothed-random} extended our approach to the case where $k$ is odd. Surprisingly, they showed that the construction we used for odd-order tensor PCA in Section~\ref{sec:cert-odd} does not work as is, and requires an additional ``row pruning'' step.

\subsubsection{Setup}

Let $x_1,\ldots,x_n$ be $\{\pm 1\}$-valued variables. A $k$-XOR formula $\Phi$ with $m$ constraints is specified by a sequence of subsets $U_1,\ldots,U_m$ with $U_i \subseteq [n]$ and $|U_i| = k$, along with values $b_1,\ldots,b_m$ with $b_i \in \{\pm 1\}$. For $1 \le i \le m$, constraint $i$ is satisfied if $x^{U_i} = b_i$, where $x^{U_i} \defeq \prod_{j \in U_i} x_j$. We write $P_\Phi(x)$ for the number of constraints satisfied by $x$. We will consider a uniformly random $k$-XOR formula in which each $U_i$ is chosen uniformly and independently from the ${n \choose k}$ possible $k$-subsets, and each $b_i$ is chosen uniformly and independently from $\{\pm 1\}$.

Given a formula $\Phi$, the goal of strong refutation is to certify an upper bound on the number of constraints that can be satisfied. In other words, our algorithm should output a bound $B = B(\Phi)$ such that for \emph{every} formula $\Phi$, $\max_{x \in \{\pm 1\}^n} P_\Phi(x) \le B(\Phi)$. (Note that this must \emph{always} be satisfied, not merely with high probability.) At the same time, we want the bound $B$ to be small with high probability over a random $\Phi$. Since a random assignment $x$ will satisfy roughly half the constraints, the best bound we can hope for is $B = \frac{m}{2}(1+\varepsilon)$ with $\varepsilon > 0$ small.

\subsubsection{Algorithm}

Let $k \ge 2$ be even and let $\ell \ge k/2$. Given a $k$-XOR formula $\Phi$, construct the order-$\ell$ symmetric difference matrix $\bM \in \RR^{{[n] \choose \ell} \times {[n] \choose \ell}}$ as follows. For $S,T \subseteq [n]$ with $|S| = |T| = \ell$ and $i \in [n]$, let
$$\bM_{S,T}^{(i)} = \left\{\begin{array}{ll} b_i & \text{if } S \symd T = U_i \\ 0 & \text{otherwise} \end{array} \right.$$
and let
$$\bM = \sum_{i=1}^m \bM^{(i)}.$$

\noindent Define the parameter
$$d_\ell \defeq {n-\ell \choose k/2}{\ell \choose k/2},$$
which, for any fixed $|S| = \ell$, is the number of sets $|T| = \ell$ such that $|S \symd T| = k$. For an assignment $x \in \{\pm 1\}^n$, let $u^x \in \RR^{[n] \choose \ell}$ be defined by $u^x_S = x^S$ for all $|S| = \ell$. We have
$$\|\bM\| \ge \frac{(u^x)^\top \bM u^x}{\|u^x\|^2} = {n \choose \ell}^{-1} \sum_{i=1}^m \sum_{S \symd T = U_i} x^{U_i} b_i = d_\ell {n \choose k}^{-1} (2P_\Phi(x) - m)$$
since for any fixed $U_i$ (of size $k$), the number of $(S,T)$ pairs such that $S \symd T = U_i$ is ${n \choose \ell} d_\ell {n \choose k}^{-1}$. Thus, we can perform strong refutation by computing $\|\bM\|$:
\begin{equation}\label{eq:M-cert}
P_\Phi(x) \le \frac{m}{2} + \frac{1}{2 d_\ell} {n \choose k} \|\bM\|.
\end{equation}

\begin{theorem}\label{thm:xor-even}
Let $k \ge 2$ be even and let $k/2 \le \ell \le n-k/2$. Let $\beta \in (0,1)$. If
\begin{equation}\label{eq:cond-m}
m \ge \frac{4 e^2 {n \choose k} \log{n \choose \ell}}{\beta^2 d_\ell}
\end{equation}
then $\|\bM\|$ certifies
$$P_\Phi(x) \le \frac{m}{2}(1+\beta)$$
with probability at least $1 - 3{n \choose \ell}^{-1}$ over a uniformly random $k$-XOR formula $\Phi$ with $m$ constraints.
\end{theorem}
\noindent If $k$ is constant and $\ell = n^\delta$ with $\delta \in (0,1)$, the condition~\eqref{eq:cond-m} becomes
$$m \ge O(\beta^{-2} n^{k/2} \ell^{1-k/2} \log n) = O(\beta^{-2} n^{k/2 + \delta(1-k/2)} \log n),$$
matching the result of \cite{strongly-refuting}. In fact, our result is tighter by polylog factors.

\subsubsection{Binomial Tail Bound}

The main ingredients in the proof of Theorem~\ref{thm:xor-even} will be the matrix Chernoff bound (Theorem~\ref{thm:matrix-conc}) and the following standard binomial tail bound.

\begin{proposition}\label{prop:binom-tail}
Let $X \sim \mathrm{Binomial}(n,p)$. For $p < \frac{u}{n} < 1$,
$$\PP\left(X \ge u\right) \le \exp\left[-u \left(\log\left(\frac{u}{pn}\right) - 1\right)\right].$$
\end{proposition}
\begin{proof}
We begin with the standard Binomial tail bound \cite{binom-tail}
$$\PP\left(X \ge u\right) \le \exp\left(-n\, D\left(\frac{u}{n}\, \Big\|\, p\right)\right)$$
for $p < \frac{u}{n} < 1$, where
$$D(a \, \| \, p) \defeq a \log\left(\frac{a}{p}\right) + (1-a) \log\left(\frac{1-a}{1-p}\right).$$
Since $\log(x) \ge 1 - 1/x$,
$$(1-a) \log\left(\frac{1-a}{1-p}\right) \ge (1-a)\left(1 - \frac{1-p}{1-a}\right) = p - a \ge -a,$$
and the desired result follows.
\end{proof}

\subsubsection{Proof}

\begin{proof}[Proof of Theorem~\ref{thm:xor-even}]
We need to bound $\|\bM\|$ with high probability over a uniformly random $k$-XOR formula $\Phi$. First, fix the subsets $U_1,\ldots,U_m$ and consider the randomness of the signs $b_i$. We can write $\bM$ as a Rademacher series
$$\bM = \sum_{i=1}^m b_i \bA^{(i)}$$
where
$$\bA^{(i)}_{S,T} = \One_{S \symd T = U_i}.$$
By the matrix Chernoff bound (Theorem~\ref{thm:matrix-conc}),
$$\PP\left(\|\bM\| \ge t\right) \le 2 {n \choose \ell} e^{-t^2/2\sigma^2} = 2 \exp\left(\log{n \choose \ell} - \frac{t^2}{2\sigma^2}\right)$$
where
$$\sigma^2 = \left\|\sum_{i=1}^m (\bA^{(i)})^2\right\|.$$
In particular,
\begin{equation}\label{eq:bound-op-M}
\PP\left(\|\bM\| \ge 2\sqrt{\sigma^2 \log{n \choose \ell}}\right) \le 2{n \choose \ell}^{-1}.
\end{equation}

\noindent Now we will bound $\sigma^2$ with high probability over the random choice of $U_1,\ldots,U_m$. We have
$$\sum_{i=1}^m (\bA^{(i)})^2 = \mathrm{diag}(D)$$
where $D_S$ is the number of $i$ for which $|S \symd U_i| = \ell$. This means $\sigma^2 = \max_{|S| = \ell} D_S$. For fixed $S \subseteq [n]$ with $|S| = \ell$, the number of sets $U \subseteq [n]$ with $|U| = k$ such that $|S \symd U| = \ell$ is $d_\ell$ and so $D_S \sim \mathrm{Binomial}\left(m,p\right)$ with $p \defeq d_\ell {n \choose k}^{-1}$. Using the Binomial tail bound (Proposition~\ref{prop:binom-tail}) and a union bound over $S$,
$$\PP\left(\sigma^2 \ge u\right) \le {n \choose \ell} \exp\left[-u \left(\log\left(\frac{u}{pm}\right) - 1\right)\right] = \exp\left[\log{n \choose \ell} - u \left(\log\left(\frac{u}{pm}\right) - 1\right)\right].$$
Provided
\begin{equation}\label{eq:2-cond}
\frac{u}{pm} \ge e^2 \quad \text{and} \quad u \ge 2 \log{n \choose \ell},
\end{equation}
we have
$$\PP\left(\sigma^2 \ge u\right) \le {n \choose \ell}^{-1}.$$

\noindent Let $\beta \in (0,1)$. From~\eqref{eq:M-cert}, to certify $P_\Phi(x) \le \frac{m}{2}(1+\beta)$ it suffices to have $\|\bM\| \le \beta m d_\ell {n \choose k}^{-1} = \beta pm$. Therefore, from~\eqref{eq:bound-op-M}, it suffices to have
$$\sigma^2 \le \frac{\beta^2 p^2 m^2}{4 \log{n \choose \ell}}.$$
From~\eqref{eq:2-cond}, this will occur provided
\begin{equation}\label{eq:final-cond-1}
\frac{\beta^2 p^2 m^2}{4 \log{n \choose \ell}} \ge e^2 pm \;\Leftrightarrow\; m \ge \frac{4e^2 \log{n \choose \ell}}{\beta^2 p}
\end{equation}
and
\begin{equation}\label{eq:final-cond-2}
\frac{\beta^2 p^2 m^2}{4 \log{n \choose \ell}} \ge 2 \log{n \choose \ell} \;\Leftrightarrow\; m \ge \frac{2\sqrt{2} \log{n \choose \ell}}{\beta p}.
\end{equation}
Note that~\eqref{eq:final-cond-2} is subsumed by~\eqref{eq:final-cond-1}. This completes the proof.
\end{proof}

\subsection{Odd-Order Tensors}
\label{sec:cert-odd}

We return now to tensor PCA rather than $k$-XOR. When the tensor order $p$ is odd, the Kikuchi Hessian suggests a natural algorithm for tensor PCA (see Appendix~\ref{app:odd}) but we are unfortunately unable to give a tight analysis of it. Here we present a related algorithm for which we are able to give a better analysis, matching SOS. The idea of the algorithm is to use a construction from the SOS literature that transforms an order-$p$ tensor (with $p$ odd) into an order-$2(p-1)$ tensor via the Cauchy--Schwarz inequality~\cite{coja-sat}. We then apply a variant of our symmetric difference matrix to the resulting even-order tensor. A similar construction was given independently in the recent work \cite{hastings-quantum} and shown to give optimal performance for all $\ell \le n^\delta$ for a certain constant $\delta > 0$. The proof we give here applies to the full range of $\ell$ values: $\ell \ll n$. Our proof uses a certain variant of the matrix Bernstein inequality combined with some fairly simple moment calculations.

\subsubsection{Setup}

For simplicity, we consider the following certification version of the tensor PCA problem. Let $p \ge 3$ be odd and let $\bY \in (\RR^n)^{\otimes p}$ be an asymmetric tensor with i.i.d.\ Rademacher (uniform $\pm 1$) entries. Our goal is to certify an upper bound on the \emph{Rademacher injective norm}, defined as
$$\|\bY\|_\pm \defeq \max_{x \in \{\pm 1\}^n/\sqrt{n}} |\langle \bY,x^{\otimes p}\rangle|.$$
The true value is $O(\sqrt{n})$ with high probability. In time $n^\ell$ (where $\ell = n^\delta$ with $\delta \in (0,1)$) we will certify the bound $\|\bY\|_\pm \le n^{p/4} \ell^{1/2 - p/4} \mathrm{polylog}(n)$, matching the results of \cite{mult-approx,cert-tensor}. Following Lemma~4.4 of \cite{HSS}, such certification results can be transformed (via sum-of-squares) into recovery results for tensor PCA under the optimal condition $\lambda \ge \ell^{1/2 - p/4}n^{-p/4} \,\mathrm{polylog}(n)$. However, our specific result would apply to a strange variant of tensor PCA where both the signal and noise are Rademacher. To handle Gaussian noise, one would need to take $Y$ Gaussian in our result, but we do not attempt this here. To certify a bound on the \emph{injective norm} instead of the Rademacher injective norm (where $x$ is constrained to the sphere instead of the hypercube), one should use the basis-invariant version of the symmetric difference matrices given by \cite{hastings-quantum} (but again, we do not attempt this here).

\subsubsection{Algorithm}

We will use a trick from \cite{coja-sat} which is often used in the sum-of-squares literature. For any $\|x\|=1$, we have by the Cauchy--Schwarz inequality,
$$\langle \bY,x^{\otimes p} \rangle^2 \le \|x\|^2 \langle \bT,x^{\otimes 4q}\rangle = \langle \bT,x^{\otimes 4q}\rangle$$
where $p = 2q+1$ and $\bT_{abcd} \defeq \sum_{e \in [n]} \bY_{ace} \bY_{bde}$ where $a,b,c,d \in [n]^q$. We have $\EE[\bT]_{abcd} = n \cdot \bone\{ac = bd\}$ and so $\langle \EE[\bT],x^{\otimes 4q} \rangle = n \sum_{ac} (x^a x^c)^2 = n$. Let $\widetilde\bT = \bT - \EE[\bT]$, i.e., $\widetilde\bT_{abcd} = \bT_{abcd} \cdot \bone\{ac \ne bd\}$. (The symbols $ab$ and $cd$ in the indicators are interpreted as concatenation of strings.)
For some $\ell \ge 2q$, define the $n^\ell \times n^\ell$ matrix $\bM$ as follows. For $S,T \in [n]^\ell$,
$$\bM_{S,T} \defeq \sum_{abcd} \widetilde\bT_{abcd}\, N_{ab,cd}^{-1} \cdot \bone\{S \lra T\}$$
where $S \lra T$ roughly means that $S$ is obtained from $T$ by replacing $ab$ by $cd$, or $cd$ by $ab$; the formal definition is given below. Also, $N_{ab,cd}$ denotes the number of $(S,T)$ pairs for which $S \lra T$.

\begin{definition}
For $S,T \in [n]^\ell$ and $a,b,c,d \in [n]^q$, we write $S \lra T$ if there are distinct indices $i_1,\ldots,i_{2q} \in [\ell]$ such that either: (i) $S_{i_j} = (ab)_j$ and $T_{i_j} = (cd)_j$ for all $j \in [2q]$, the values in $a,b,c,d$ do not appear anywhere else in $S$ or $T$, and $S,T$ are identical otherwise: $S_i = T_i$ for all $i \notin \{i_1,\ldots,i_{2q}\}$; or (ii) the same holds but with $ab$ and $cd$ interchanged.
\end{definition}

\noindent Note that
\begin{equation}\label{eq:N-bar}
N_{ab,cd} \ge \overline{N} := {\ell \choose 2q} (n-4q)^{\ell-2q}.
\end{equation}

\noindent The above construction ensures that
$$n^\ell (x^{\otimes \ell})^\top \bM (x^{\otimes \ell}) = n^{2q} \langle \widetilde\bT, x^{\otimes 4q} \rangle \qquad\text{for all } x \in \{\pm 1\}^n/\sqrt{n}.$$
This means we can certify an upper bound on $\|\bY\|_\pm$ by computing $\|\bM\|$:
$$\|\bY\|_\pm \le \sqrt{\langle \bT,x^{\otimes 4q} \rangle} \le \sqrt{\langle \EE[\bT],x^{\otimes 4q} \rangle} + \sqrt{\langle \widetilde\bT,x^{\otimes 4q} \rangle} \le n^{1/2} + n^{\ell/2-q} \|\bM\|^{1/2}.$$

\begin{theorem}\label{thm:cert-odd}
Let $p \ge 3$ be odd and let $p-1 \le \ell \le \min\{\frac{n}{4(p-1)},\frac{n}{8 \log n}\}$. Then $\|\bM\|$ certifies
$$\|\bY\|_\pm \le n^{1/2} + 8 p^p \ell^{1/2-p/4} n^{p/4} (\log n)^{1/4}$$
with probability at least $1 - n^{-\ell}$ over an i.i.d.\ Rademacher $\bY$.
\end{theorem}

\subsubsection{Proof}

\noindent We will use the following variant of the matrix Bernstein inequality; this is a special case ($\bA_k = R \cdot \bI$) of \cite{tropp}, Theorem~6.2.
\begin{theorem}[Matrix Bernstein]
\label{thm:matrix-bernstein}
Consider a finite sequence $\{\bX_i\}$ of independent random symmetric $d \times d$ matrices. Suppose $\EE[\bX_i] = 0$ and $\|\EE[\bX_i^r]\| \le \frac{r!}{2} R^r$ for $r = 2,3,4,\ldots$. Then
$$\Pr\left\{\left\|\sum_{i=1}^n \bX_i\right\| \ge t\right\} \le d \cdot \exp\left(\frac{-t^2/2}{nR^2 + Rt}\right).$$
\end{theorem}

\noindent For $e \in [n]$, let
$$\bM^{(e)}_{S,T} \defeq \sum_{abcd} \bY_{ace}\bY_{bde}\, N_{ab,cd}^{-1} \cdot \bone\{S \lra T\} \cdot \bone\{ac \ne bd\}.$$
We will apply Theorem~\ref{thm:matrix-bernstein} to the sum $\bM = \sum_e \bM^{(e)}$. Note that $\EE[\bM^{(e)}] = 0$. To bound the moments $\|\EE[(\bM^{(e)})^r]\|$, we will use the following basic fact.

\begin{lemma}\label{lem:row-fact}
If $\bA$ is a symmetric matrix,
\[ \|\bA\| \le \max_j \sum_i |\bA_{ij}|. \]
\end{lemma}
\begin{proof}
Let $v$ be the leading eigenvector of $\bA$ so that $\bA v = \lambda v$ where $\|\bA\| = |\lambda|$. Normalize $v$ so that $\|v\|_1 = 1$. Then $\|\bA v\|_1 = |\lambda| \cdot \|v\|_1$ and so
\[ \|\bA\| = |\lambda| \cdot \|v\|_1 = \sum_i \left|\sum_j \bA_{ij} v_j\right| \le \sum_{ij} |A_{ij}| \cdot |v_j| \le \sum_j |v_j| \cdot \sum_i |\bA_{ij}| \le \|v\|_1 \cdot \max_j \sum_i |\bA_{ij}|. \qedhere\]
\end{proof}

\begin{proof}[Proof of Theorem~\ref{thm:cert-odd}]
For any fixed $e$, we have by Lemma~\ref{lem:row-fact},
\[ \|\EE[(\bM^{(e)})^r]\| \le \max_S \sum_T |\EE[(\bM^{(e)})^r]_{S,T}| =: \max_S h(r,e,S). \]
Let $\pi$ denote a ``path'' of the form
$$\pi = (S_0, a_1,b_1,c_1,d_1, S_1, a_2,b_2,c_2,d_2, S_2, \ldots, a_r,b_r,c_r,d_r, S_r)$$
such that $S_0 = S$, $(a_i,c_i) \ne (b_i,d_i)$, and $S_{i-1} \lra[a_i b_i, c_i d_i] S_i$. Then we have
$$h(r,e,S) = \sum_\pi \EE \prod_{i=1}^r \bY_{a_i c_i e} \bY_{b_i d_i e} N^{-1}_{a_i b_i, c_i d_i}.$$

\noindent Among tuples of the form $(a_i,c_i)$ and $(b_i,d_i)$, each must occur an even number of times (or else the term associated with $\pi$ is $0$). There are $2r$ such tuples, so there are ${2r \choose r}\, r!\, 2^{-r}$ ways to pair them up (via a perfect matching on $2r$ elements). Once $S_{i-1}$ is chosen, there are at most $2(\ell n)^q$ choices for $(a_i, c_i)$---where the factor of 2 accounts for the possibility of interchanging $ab$ and $cd$---and the same is true for $(b_i, d_i)$. However, the second occurrence of $(a_i,c_i)$ or $(b_i,d_i)$ within a pair has only 1 choice. Once $S_{i-1}, a_i, b_i, c_i, d_i$ are chosen, there are at most $(2q)!$ possible choices for $S_i$ (the worst case being e.g.\ when $a_i,b_i$ are all 1's, but $c_i,d_i$ have distinct elements). This means

\[ h(r,e,S) \le {2r \choose r} \,r!\,2^{-r} [2(\ell n)^q \cdot (2q)!]^r \overline{N}^{-r}. \]
where $\overline{N}$ is defined in~\eqref{eq:N-bar}. Since ${2r \choose r} \le 4^r$, we can apply Theorem~\ref{thm:matrix-bernstein} with $R = 8 (2q)! (\ell n)^q \overline{N}^{-1}$. This yields
$$\Pr\left\{\|\bM\| \ge t\right\} \le n^\ell \cdot \exp\left(\frac{-t^2/2}{nR^2 + Rt}\right).$$

\noindent Let $t = R \sqrt{8 \ell n\log n}$. Provided $\ell \le n/(8 \log n)$ we have $Rt \le nR^2$ and so
$$\Pr\left\{\|\bM\| \ge R \sqrt{8 \ell n \log n}\right\} \le \exp\left(\ell \log n - \frac{t^2}{4nR^2}\right) = \exp(-\ell \log n) = n^{-\ell}.$$

\noindent Thus, with high probability we certify
\begin{align*}
\|\bY\|_\pm &\le n^{1/2} + n^{\ell/2-q} \|\bM\|^{1/2} \\
&\le n^{1/2} + n^{\ell/2-q} R^{1/2} (8 \ell n \log n)^{1/4} \\
&= n^{1/2} + n^{\ell/2-q} \sqrt{8(2q)!}\, (\ell n)^{q/2} \overline{N}^{-1/2} (8 \ell n \log n)^{1/4} \\
&= n^{1/2} + 8^{3/4} \sqrt{(2q)!}\, n^{\ell/2-q} \overline{N}^{-1/2} (\ell n)^{1/4+q/2} (\log n)^{1/4}.
\end{align*}

\noindent We have the following bound on $\overline{N}$:
\begin{align*}
\overline{N} &= {\ell \choose 2q}(n-4q)^{\ell-2q} \\
&= n^\ell \cdot \frac{{\ell \choose 2q}(n-4q)^{\ell-2q}}{n^\ell} \\
&\ge n^\ell \cdot \frac{\ell^{2q}}{(2q)^{2q}} \cdot \frac{(n-4q)^{\ell-2q}}{n^\ell} \\
&= \frac{n^\ell}{(2q)^{2q}} \cdot \left(\frac{n-4q}{n}\right)^{\ell-2q} \left(\frac{\ell}{n}\right)^{2q} \\
&= \frac{n^\ell}{(2q)^{2q}} \cdot \left(1 - \frac{4q}{n}\right)^{\ell-2q} \left(\frac{\ell}{n}\right)^{2q} \\
&\ge \frac{n^\ell}{p^p} \cdot \left(1 - (\ell-2q)\frac{4q}{n}\right) \left(\frac{\ell}{n}\right)^{2q} \\
&\ge \frac{n^\ell}{p^p} \cdot \left(1 - \frac{4q \ell}{n}\right) \left(\frac{\ell}{n}\right)^{2q} \\
&\ge \frac{n^\ell}{2p^p} \left(\frac{\ell}{n}\right)^{2q}
\end{align*}
provided $\ell \le n/(8q)$, i.e., $\ell \le n/[4(p-1)]$. Therefore, we certify
\begin{align*}
\|\bY\|_\pm &\le n^{1/2} + 2^{1/2} \cdot 8^{3/4} \sqrt{(2q)!} \, p^{p/2} \ell^{1/4+q/2} n^{1/4-q/2} (\ell/n)^{-q} (\log n)^{1/4} \\
&\le n^{1/2} + 8 p^p \ell^{1/2-p/4} n^{p/4} (\log n)^{1/4},
\end{align*}
as desired.
\end{proof}

\section{Conclusion}

We have presented a hierarchy of spectral algorithms for tensor PCA, inspired by variational inference and statistical physics. In particular, the core idea of our approach is to locally minimize the Kikuchi free energy. We specifically implemented this via the Kikuchi Hessian, but there may be many other viable approaches to minimizing the Kikuchi free energy such as generalized belief propagation \cite{GBP}. Broadly speaking, we conjecture that for many average-case problems, algorithms based on Kikuchi free energy and algorithms based on sum-of-squares should both achieve the optimal tradeoff between runtime and statistical power. One direction for further work is to verify that this analogy holds for problems other than tensor PCA; in particular, we show here that it also applies to refuting random $k$-XOR formulas when $k$ is even. Some other models to consider in the future could be planted clique or sparse PCA, which also have smooth tradeoffs between SNR and runtime~\cite{alon-clique,subexp-sparse}.

Perhaps one benefit of the Kikuchi hierarchy over the sum-of-squares hierarchy is that it has allowed us to \emph{systematically} obtain spectral methods, simply by computing a certain Hessian matrix. Furthermore, the algorithms we obtained are simpler than their SOS counterparts. We are hopeful that the Kikuchi hierarchy will provide a roadmap for systematically deriving simple and optimal algorithms for a large class of problems.

\section*{Acknowledgments}

We thank Alex Russell for suggesting the matrix Chernoff bound (Theorem~\ref{thm:matrix-conc}). For helpful discussions, we thank Afonso Bandeira, Sam Hopkins, Pravesh Kothari, Florent Krzakala, Tselil Schramm, Jonathan Shi, and Lenka Zdeborov\'a. This project started during the workshop \emph{Spin Glasses and Related Topics} held at the Banff International Research Station (BIRS) in the Fall of 2018. We thank our hosts at BIRS as well as the workshop organizers: Antonio Auffinger, Wei-Kuo Chen, Dmitry Panchenko, and Lenka Zdeborov\'a.

\bibliographystyle{alpha}
\bibliography{main}

\appendix

\section{Detection for General Priors} \label{app:all-priors}

While we have mainly focused on the Rademacher-spiked tensor model, we now show that our algorithm works just as well (at least for detection) for a much larger class of spike priors.

\begin{theorem}
Let $p \ge 2$ be even. Consider the spiked tensor model with a spike prior $P_{\tx}$ that draws the entries of $x_*$ i.i.d.\ from some distribution $\pi$ on $\RR$ (which does not depend on $n$), normalized so that $\EE[\pi^2] = 1$. There is a constant $C$ (depending on $p$ and $\pi$) such that if $\lambda \ge C \ell^{1/2} d_\ell^{-1/2} \sqrt{\log n}$ then Algorithm~\ref{detection_even} achieves strong detection.
\end{theorem}

\begin{proof}
From \eqref{eq:Z-bound}, we have $\|\bZ\|_{\textup{op}} = O(\sqrt{\ell d_\ell \log n})$ with high probability, and so it remains to give a lower bound on $\|\bX\|_{\textup{op}}$. Letting $u_S = \prod_{i \in S} \sgn((x_*)_i)$ for $|S|=\ell$,
$$\|\bX\|_{\textup{op}} \ge \frac{u^\top \bX u}{\|u\|^2} = {n \choose \ell}^{-1} u^\top \bX u$$
where
\begin{align*}
u^\top \bX u
&= \sum_{|S \symd T|=p} u_S \bX_{S,T} u_T \\
&= \sum_{|S \symd T|=p} x_*^{S \symd T} \prod_{i \in S} \sgn((x_*)_i) \prod_{i \in T} \sgn((x_*)_i) \\
&= \sum_{|S \symd T|=p} \left|x_*^{S \symd T}\right|.
\end{align*}
We have
\begin{equation}\label{eq:mean-uXu}
\EE[u^\top \bX u] = {n \choose \ell} d_\ell \left(\EE|\pi|\right)^p = C(\pi,p) {n \choose \ell} d_\ell,
\end{equation}
and
\begin{align}
\mathrm{Var}\left[u^\top \bX u\right] &= \mathrm{Var}\left[\sum_{|S \symd T|=p} \left|x_*^{S \symd T}\right|\right] \nonumber\\
&= \sum_{|S \symd T|=p} \,\sum_{|S' \symd T'|=p} \mathrm{Cov}(|x_*^{S \symd T}|,|x_*^{S' \symd T'}|).
\label{eq:var-uXu}
\end{align}
We have
\begin{align*}
\mathrm{Cov}(|x_*^{S \symd T}|,|x_*^{S' \symd T'}|) 
&\le \sqrt{\mathrm{Var}\left(|x_*^{S \symd T}|\right)\mathrm{Var}\left(|x_*^{S' \symd T'}|\right)}\\ 
&= \mathrm{Var}\left(|x_*^{S \symd T}|\right) 
\le \EE\left[|x_*^{S \symd T}|^2\right] 
= (\EE[\pi^2])^p = 1 \, . 
\end{align*}
Also, $\mathrm{Cov}(|x_*^{S \symd T}|,|x_*^{S' \symd T'}|) = 0$ unless $S \symd T$ and $S' \symd T'$ have nonempty intersection. Using Lemma~\ref{lem:intersect} (below), the fraction of terms in \eqref{eq:var-uXu} that are nonzero is at most $p^2/n$ and so
\begin{equation}\label{eq:var-uXu-bound}
\mathrm{Var}\left[u^\top \bX u\right] \le \left[{n \choose \ell}d_\ell\right]^2 \frac{p^2}{n}.
\end{equation}
By Chebyshev's inequality, it follows from \eqref{eq:mean-uXu} and \eqref{eq:var-uXu-bound} that $u^\top \bX u \ge \frac{1}{2} C(\pi, p) {n \choose \ell}d_{\ell}$ with probability at least $1- \frac{4p^2}{C(\pi,p)^2 n}$. This implies $\|\bX\|_{\textup{op}} \ge \frac{1}{2} C(\pi, p) d_\ell$ with the same probability, and so we have strong detection provided $\lambda \ge c_0 \ell^{1/2} d_\ell^{-1/2} \sqrt{\log n}$ for a particular constant $c_0 = c_0(\pi,p)$.
\end{proof}

Above, we made use of the following lemma.
\begin{lemma}\label{lem:intersect}
Fix $A \subseteq [n]$ with $|A| = a$. Let $B$ be chosen uniformly at random from all subsets of $[n]$ of size $b$. Then $\PP(A \cap B \ne \varnothing) \le \frac{ab}{n}$.
\end{lemma}
\begin{proof}
Each element of $A$ will lie in $B$ with probability $b/n$, so the result follows by a union bound over the elements of $A$.
\end{proof}

\section{The Odd-$p$ Case} \label{app:odd}

When $p$ is odd, the Kikuchi Hessian still gives rise to a spectral algorithm. While we conjecture that this algorithm is optimal, we unfortunately only know how to prove suboptimal results for it. (However, we can prove optimal results for a related algorithm; see Section~\ref{sec:cert-odd}.) We now state the algorithm and its conjectured performance.

Let $p$ be odd and fix an integer $\ell \in [\lfloor p/2 \rfloor, n - \lceil p/2 \rceil]$. Consider the symmetric difference matrix $\bM \in \RR^{{n \choose \ell} \times {n \choose \ell+1}}$ with entries
\[\bM_{S,T} = \left\{\begin{array}{ll} \bY_{S \symd T} & \text{if } |S \symd T| = p, \\ 0 & \text{otherwise,}\end{array}\right.\]
where $S,T \subseteq [n]$ with $|S| = \ell$ and $|T| = \ell+1$.

\begin{algorithm}[Recovery for odd $p$] \label{alg:recovery-odd}
Let $\bu$ be a (unit-norm) top left-singular vector of $\bM$ and let $\bv = \bM^\top \bu$ be the corresponding top right-singular vector. Output $\widehat{\bx} = \widehat{\bx}(\bY) \in \RR^n$, defined by
\[\widehat{\bx}_i = \sum_{S \in {[n] \choose \ell}, \,T \in {[n] \choose \ell+1}} \bu_S \bv_T \bone_{S \symd T = \{i\}}, \quad i \in [n].\]
\end{algorithm}

Notice that the rounding step consisting in extracting an $n$-dimensional vector $\widehat{\bx}$ from the singular vectors of $\bM$ is slightly simpler that the even-$p$ case, in that it does not require forming a voting matrix.
We conjecture that, like the even case, this algorithm matches the performance of SOS.

\begin{conjecture}
Consider the Rademacher-spiked tensor model with $p \ge 3$ odd. If
\[
\lambda \gg \ell^{-(p-2)/4} n^{-p/4} 
\]
then (i) there is a threshold $\tau = \tau(n,p,\ell,\lambda)$ such that strong detection can be achieved by thresholding the top singular value of $\bM$ at $\tau$, and (ii) Algorithm~\ref{alg:recovery-odd} achieves strong recovery.
\end{conjecture}

Similarly to the proof of Theorem~\ref{recovery_result}, the matrix Chernoff bound (Theorem~\ref{thm:matrix-conc}) can be used to show that strong recovery is achievable when $\lambda \gg \ell^{-(p-1)/4} n^{-(p-1)/4}$, which is weaker than SOS when $\ell \ll n$. We now explain the difficulties involved in improving this. We can decompose $\bM$ into a signal part and a noise part: $\bM = \lambda \bX + \bZ$. In the regime of interest, $\ell^{-(p-2)/4} n^{-p/4} \ll \lambda \ll \ell^{-(p-1)/4} n^{-(p-1)/4}$, the signal term is smaller in operator norm than the noise term, i.e., $\lambda \|\bX\|_{\textup{op}} \ll \|\bZ\|_{\textup{op}}$. While at first sight this would seem to suggest that detection and recovery are hopeless, we actually expect that $\lambda \bX$ still affects the top singular value and singular vectors of $\bM$. This phenomenon is already present in the analysis of tensor unfolding (the case $p=3$, $\ell = 1$) \cite{HSS}, but it seems that new ideas are required to extend the analysis beyond this case.

\section{Proof of Boosting}
\label{app:boost}
In this section we prove Proposition~\ref{prop:boost} which we restate here: 

\begin{proposition}[Proposition~\ref{prop:boost} restated]\label{prop:boost0}
Let $\bY \sim \PP_{\lambda}$ with any spike prior $P_{\tx}$ supported on $\mathcal{S}^{n-1}(\sqrt{n})$. Suppose we have an initial guess $\bu \in \RR^n$ satisfying $\corr(\bu,\bx_*) \ge \tau$. Obtain $\widehat \bx$ from $\bu$ via a single iteration of the tensor power method: $\widehat \bx = \bY\{\bu\}$. There exists a constant $c = c(p)>0$ such that with high probability, 
\[\corr(\widehat \bx,\bx_*) \ge 1 - c\lambda^{-1} \tau^{1-p} n^{(1-p)/2}.\] 
In particular, if $\tau > 0$ is any constant and $\lambda = \omega(n^{(1-p)/2})$ then $\corr(\widehat \bx,\bx) = 1-o(1)$.
\end{proposition}

\begin{definition}
For a tensor $\bG \in (\RR^n)^{\otimes p}$, the \emph{injective tensor norm} is
$$\|\bG\|_{\textup{inj}} \defeq \max_{\|u^{(1)}\|= \cdots =\|u^{(p)}\|=1} \sum_{i_1,\ldots,i_p} \bG_{i_1,\ldots,i_p} u^{(1)}_{i_1} \cdots u^{(p)}_{i_p},$$
where $u^{(j)} \in \RR^n$. For a symmetric tensor $\bG$, it is known \cite{wat} that equivalently,
$$\|\bG\|_{\textup{inj}} = \max_{\|u\|=1} \left|\sum_{i_1,\ldots,i_p} \bG_{i_1,\ldots,i_p} u_{i_1} \cdots u_{i_p}\right|.$$
\end{definition}

\begin{proof}[Proof of Proposition~\ref{prop:boost}]
Write $\hat x = \lambda \langle u,x_* \rangle^{p-1} x_* + \Delta$ where $\|\Delta\| \le \|\bG\|_{\textup{inj}} \|u\|^{p-1}$.
We have
$$|\langle \hat x,x_* \rangle| \ge \lambda |\langle u,x \rangle|^{p-1} \|x_*\|^2 - \|\Delta\| \|x_*\|,$$
and
$$\|\hat x\| \le \lambda |\langle u,x_* \rangle|^{p-1} \|x_*\| + \|\Delta\|,$$
and so
\begin{align*}
\corr(\hat x,x_*) &= \frac{|\langle \hat x,x_* \rangle|}{\|\hat x\| \|x_*\|}\\
&\ge \frac{\lambda |\langle u,x_* \rangle|^{p-1} \|x_*\| - \|\Delta\|}{\lambda |\langle u,x_* \rangle|^{p-1} \|x_*\| + \|\Delta\|}\\
&= 1 - \frac{2 \|\Delta\|}{\lambda |\langle u,x_* \rangle|^{p-1} \|x_*\| + \|\Delta\|}\\
&\ge 1 - \frac{2 \|\Delta\|}{\lambda |\langle u,x_* \rangle|^{p-1} \|x_*\|}\\
&\ge 1 - \frac{2 \|\bG\|_{\textup{inj}} \|u\|^{p-1}}{\lambda |\langle u,x_* \rangle|^{p-1} \|x_*\|}\\
&\ge 1 - \frac{2 \|\bG\|_{\textup{inj}}}{\lambda \tau^{p-1} \|x_*\|^p}.
\end{align*}
Our prior $P_{\tx}$ is supported on the sphere of radius $\sqrt{n}$, so $\|x_*\| = \sqrt{n}$. We need to control the injective norm of the tensor $\bG$. To this end we use Theorem 2.12 in~\cite{ABC} (see also Lemma~2.1 of~\cite{RM-tensor}): there exists a constant $c(p)>0$ (called $E_0(p)$ in~\cite{ABC}) such that for all $\epsilon > 0$,
\[\PP\Big(\sqrt{\frac{p}{n}}\|\bG\|_{\textup{inj}}  \ge c(p) + \epsilon\Big) \longrightarrow 0 \quad \text{as } n \to \infty.\]
Letting $\epsilon = c(p)$ we obtain 
\[\corr(\hat x,x_*) \ge 1- 4\frac{c(p)}{\sqrt{p}}\frac{n^{(1-p)/2}}{\lambda \tau^{p-1}},\]
with probability tending to 1 as $n \to \infty$.
\end{proof}

\section{Computing the Kikuchi Hessian}

In Section~\ref{sec:motivation} we defined the Kikuchi free energy and explained the high level idea of how the symmetric difference matrices are derived from the Kikuchi Hessian. We now carry out the Kikuchi Hessian computation in full detail. This is a heuristic (non-rigorous) computation, but we believe these methods are important as we hope they will be useful for systematically obtaining optimal spectral methods for a wide variety of problems.

\label{app:hessian}
\subsection{Derivatives of Kikuchi Free Energy}

Following~\cite{bethe-hessian}, we parametrize the beliefs in terms of the moments $m_S = \mathbb{E}[x^S]$. Specifically,
\begin{equation}\label{eq:bm}
b_S(x_S) = \frac{1}{2^{|S|}} \left(1 + \sum_{\varnothing \subset T \subseteq S} m_T x^T\right).
\end{equation}
We imagine $m_T$ are close enough to zero so that $b_S$ is a positive measure. One can check that these beliefs indeed have the prescribed moments: for $T \subseteq S$,
$$\sum_{x_S} b_S(x_S) x^T = m_T.$$
Thus we can think of the Kikuchi free energy $\mathcal{K}$ as a function of the moments $\{m_S\}_{0 < |S| \le r}$. This parametrization forces the beliefs to be consistent, i.e., if $T \subseteq S$ then the marginal distribution $b_S|T$ is equal to $b_T$.

We now compute first and second derivatives of $\mathcal{K} = \mathcal{E} - \mathcal{S}$ with respect to the moments $m_S$. First, the energy term:
$$\mathcal{E} = -\lambda \sum_{|S|=p} \bY_S m_S$$
$$\frac{\partial \mathcal{E}}{\partial m_S} = \left\{\begin{array}{ll} -\lambda \bY_S & \text{if }|S| = p \\ 0 & \text{otherwise}\end{array}\right.$$
$$\frac{\partial^2 \mathcal{E}}{\partial m_S \partial m_{S'}} = 0.$$

\noindent Now the entropy term:
$$\frac{\partial \mathcal{S}_S}{\partial b_S(x_S)} = -\log b_S(x_S) - 1.$$
From (\ref{eq:bm}), for $\varnothing \subset T \subseteq S$,
$$\frac{\partial b_S(x_S)}{\partial m_T} = \frac{x^T}{2^{|S|}}$$
and so
\begin{align}
\frac{\partial \mathcal{S}_S}{\partial m_T} &= \sum_{x_S} \frac{\partial \mathcal{S}_S}{\partial b_S(x_S)} \cdot \frac{\partial b_S(x_S)}{\partial m_T}\nonumber\\
&= -2^{-|S|} \sum_{x_S} x^T [\log b_S(x_S)+1]\nonumber\\
&= -2^{-|S|} \sum_{x_S} x^T \log b_S(x_S).
\label{eq:deriv-Sm}
\end{align}

\noindent For $\varnothing \subset T \subseteq S$, $\varnothing \subset T' \subseteq S$,
\begin{align*}
\frac{\partial^2 \mathcal{S}_S}{\partial m_T \partial m_{T'}} &= -2^{-|S|} \sum_{x_S} x^T b_S(x_S)^{-1} \cdot \frac{\partial b_S(x_S)}{\partial m_{T'}}\\
&= -2^{-2|S|} \sum_{x_S} x^T x^{T'} b_S(x_S)^{-1}\\
&= -2^{-2|S|} \sum_{x_S} x^{T \symd T'} b_S(x_S)^{-1} \, . 
\end{align*}
Finally, if $T \not\subseteq S$ then $\frac{\partial \mathcal{S}_S}{\partial m_T} = 0$.

\subsection{The Case $r = p$}

We first consider the simplest case, where $r$ is as small as possible: $r = p$. (We need to require $r \ge p$ in order to express the energy term in terms of the beliefs.)

\subsubsection{Trivial Stationary Point}

There is a ``trivial stationary point'' of the Kikuchi free energy where the beliefs only depend on local information. Specifically, if $|S| < p$ then $b_S$ is the uniform distribution over $\{\pm 1\}^{|S|}$, and if $|S| = p$ then
$$b_S(x_S) \propto \exp\left(\lambda \bY_S x^S \right)$$
i.e.,
\begin{equation}\label{eq:b-triv}
b_S(x_S) = \frac{1}{Z_S} \exp\left(\lambda \bY_S x^S \right)
\end{equation}
where
$$Z_S = \sum_{x_S} \exp\left(\lambda \bY_S x^S \right).$$

Note that these beliefs are consistent (if $T \subseteq S$ with $|S| \le p$ then $b_S|T = b_T$) and so there is a corresponding set of moments $\{m_S\}_{|S| \le p}$. 

We now check that this is indeed a stationary point of the Kikuchi free energy. Using (\ref{eq:deriv-Sm}) and (\ref{eq:b-triv}) we have for $\varnothing \subset T \subseteq S$ and $|S| \le p$,

\begin{align*}
\frac{\partial \mathcal{S}_S}{\partial m_T} &=-2^{-|S|} \sum_{x_S} x^T \log b_S(x_S)\\
&=-2^{-|S|} \sum_{x_S} x^T \left[- \log Z_S + \lambda \bone_{|S|=p} \bY_S x^S \right]\\
&= \left\{\begin{array}{ll} -\lambda \bY_T & \text{if } |T| = p \\ 0 & \text{otherwise.} \end{array}\right.
\end{align*}

Thus if $|T| < p$ we have $\frac{\partial \mathcal{K}}{\partial m_T} = 0$, and if $|T| = p$ we have

\begin{align*}
\frac{\partial \mathcal{K}}{\partial m_T} &= \frac{\partial}{\partial m_T} \left[\mathcal{E} - \sum_{0 < |S| \le p} c_S \mathcal{S}_S \right]\\
&= -\lambda \bY_T + c_T \lambda \bY_T\\
&= 0
\end{align*}
This confirms that we indeed have a stationary point.

\subsubsection{Hessian}

We now compute the Kikuchi Hessian, the matrix indexed by subsets $0 < |T| \le p$ with entries $\bH_{T,T'} = \frac{\partial^2 \mathcal{K}}{\partial m_T \partial m_{T'}}$, evaluated at the trivial stationary point. Similarly to the Bethe Hessian \cite{bethe-hessian}, we expect the bottom eigenvector of the Kikuchi Hessian to be a good estimate for the (moments of) the true signal. This is because this bottom eigenvector indicates the best local direction for improving the Kikuchi free energy, starting from the trivial stationary point. If all eigenvalues of $\bH$ are positive then the trivial stationary point is a local minimum and so an algorithm acting locally on the beliefs should not be able to escape from it, and should not learn anything about the signal. On the other hand, a negative eigenvalue (or even an eigenvalue close to zero) indicates a (potential) direction for improvement. 

\begin{remark}
When $p$ is odd, we cannot hope for a substantially negative eigenvalue because $\bx_*$ and $-\bx_*$ are \emph{not} equally-good solutions and so the Kikuchi free energy should be locally cubic instead of quadratic. Still, we believe that the bottom eigenvector of the Kikuchi Hessian (which will have eigenvalue close to zero) yields a good algorithm. For instance, we will see in the next section that this method yields a close variant of tensor unfolding when $r=p=3$.
\end{remark}

Recall that for $\varnothing \subset T \subseteq S$, and $\varnothing \subset T' \subseteq S$,
$$\frac{\partial^2 \mathcal{S}_S}{\partial m_T \partial m_{T'}} = -2^{-2|S|} \sum_{x_S} x^{T \symd T'} b_S(x_S)^{-1}.$$
If $|S| < p$ then $b_S$ is uniform on $\{\pm 1\}^{|S|}$ (at the trivial stationary point) and so $\frac{\partial^2 \mathcal{S}_S}{\partial m_T \partial m_{T'}} = -\One_{T=T'}$. If $|S| = p$ then $b_S(x_S) = \frac{1}{Z_S} \exp(\lambda \bY_S x^S)$ where $Z_S = \sum_{x_S} \exp(\lambda \bY_S x^S) = 2^{|S|} \cosh(\lambda \bY_S)$, and so

\begin{align*}
\frac{\partial^2 \mathcal{S}_S}{\partial m_T \partial m_{T'}} &= -2^{-2|S|} \sum_{x_S} x^{T \symd T'} b_S(x_S)^{-1}\\
&= -2^{-2|S|} \sum_{x_S} x^{T \symd T'} Z_S \exp(-\lambda \bY_S x^S)\\
&= -2^{-|S|} \sum_{x_S} x^{T \symd T'} \cosh(\lambda \bY_S) \left(1 - \lambda \bY_S x^S + \frac{1}{2!} \lambda^2 \bY_S^2 - \frac{1}{3!} \lambda^3 \bY_S^3 x^S + \cdots\right)\\
&= -2^{-|S|} \sum_{x_S} x^{T \symd T'} \cosh(\lambda \bY_S) \left(\cosh(\lambda \bY_S) - \sinh(\lambda \bY_S) x^S\right)
\intertext{since $\cosh(x) = 1 + \frac{1}{2!} x^2 + \frac{1}{4!}x^4 + \cdots$ and $\sinh(x) = x + \frac{1}{3!} x^3 + \frac{1}{5!} x^5 + \cdots$}
&= \left\{\begin{array}{ll} -\cosh^2(\lambda \bY_S) & \text{if } T=T', T \subseteq S \\ \cosh(\lambda \bY_S) \sinh(\lambda \bY_S) & \text{if } T \symd T' = S, T \cup T' \subseteq S \\ 0 & \text{otherwise} \end{array}\right.\\
&= \left\{\begin{array}{ll} -\cosh^2(\lambda \bY_S) & \text{if } T=T', T \subseteq S \\ \cosh(\lambda \bY_S) \sinh(\lambda \bY_S) & \text{if } T \sqcup T' = S \\ 0 & \text{otherwise} \end{array}\right.
\end{align*}
where $\sqcup$ denotes disjoint union. (Note that we have replaced $\symd$ with $\sqcup$ due to the restriction $T,T' \subseteq S$.) We can now compute the Hessian:

\begin{align}
\bH_{T,T'} &= \frac{\partial^2 \mathcal{K}}{\partial m_T \partial m_{T'}} \nonumber\\
\label{eq:H-sum}
&= -\sum_{\substack{S \supseteq T \cup T' \\ |S| \le p}} c_S \frac{\partial^2 \mathcal{S}_S}{\partial m_T \partial m_{T'}} \\
&= \left\{\begin{array}{ll} \sum_{\substack{S \supseteq T \\ |S| < p}} c_S + \sum_{\substack{S \supseteq T \\ |S| = p}} \cosh^2(\lambda \bY_S) & \text{if } T=T' \\ -\cosh(\lambda \bY_{T \sqcup T'}) \sinh(\lambda \bY_{T \sqcup T'}) & \text{if } |T \sqcup T'| = p \\ 0 & \text{otherwise} \end{array}\right. \nonumber\\
&= \left\{\begin{array}{ll} 1 + \sum_{\substack{S \supseteq T \\ |S| = p}} [\cosh^2(\lambda \bY_S) - 1] & \text{if } T=T' \\ -\cosh(\lambda \bY_{T \sqcup T'}) \sinh(\lambda \bY_{T \sqcup T'}) & \text{if } |T \sqcup T'| = p \\ 0 & \text{otherwise} \end{array}\right. \nonumber
\end{align}
where we used (\ref{eq:counting}) in the last step. Suppose $\lambda \ll 1$ (since otherwise tensor PCA is very easy). If $T = T'$ then, using the $\cosh$ Taylor series, we have the leading-order approximation
\begin{align*}
\bH_{T,T} &\approx 1 + \sum_{\substack{S \supseteq T \\ |S| = p}} \lambda^2 \bY_S^2\\
&\approx 1 + {n-|T| \choose p-|T|} \lambda^2 \mathbb{E}[\bY_S^2]\\
&\approx 1 + \frac{n^{p-|T|}}{(p-|T|)!} \lambda^2.
\end{align*}

Using the Taylor series for $\cosh \cdot \sinh$, this means $\bH \approx \widetilde \bH$ where
$$\widetilde \bH_{T,T'} = \left\{\begin{array}{ll} 1 \vee \frac{n^{p-|T|}}{(p-|T|)!} \lambda^2 & \text{if } T=T' \\ -\lambda \bY_{T \sqcup T'} & \text{if } |T \sqcup T'| = p \\ 0 & \text{otherwise.} \end{array}\right.$$

\subsubsection{The Case $r = p = 3$}
\label{app:kik-unfolding}

We now restrict to the case $r = p = 3$ and show that the Kikuchi Hessian recovers (a close variant of) the tensor unfolding method. Recall that in this case the computational threshold is $\lambda \sim n^{-3/4}$ and so we can assume $\lambda \ll n^{-1/2}$ (or else the problem is easy). We have
$$\widetilde \bH_{T,T'} = \left\{\begin{array}{ll} \frac{1}{2} n^2 \lambda^2 & \text{if } T=T' \text{ with } |T|=1 \\ 1 & \text{if } T=T' \text{ with } |T| \in \{2,3\} \\ -\lambda \bY_{T \sqcup T'} & \text{if } |T \sqcup T'| = 3 \\ 0 & \text{otherwise.} \end{array}\right.$$
This means we can write
$$\widetilde \bH = \left(\begin{array}{ccc} \alpha I & -\lambda \bM & 0 \\ -\lambda \bM^\top & I & 0 \\ 0 & 0 & I \end{array}\right)$$
where $\alpha = \frac{1}{2} n^2 \lambda^2$ and $\bM$ is the $n \times {n \choose 2}$ flattening of $\bY$, i.e., $\bM_{i,\{j,k\}} = \bY_{ijk} \bone_{\{i,j,k \text{ distinct}\}}$.

Since we are looking for the minimum eigenvalue of $\widetilde \bH$, we can restrict ourselves to the submatrix $\widetilde \bH^{\le 2}$ indexed by sets of size 1 and 2. We have
$$\widetilde \bH^{\le 2} = \left(\begin{array}{cc} \alpha I & -\lambda \bM \\ -\lambda \bM^\top & I \end{array}\right).$$

\noindent An eigenvector $[u \; v]^\top$ of $\widetilde \bH^{\le 2}$ with eigenvalue $\beta$ satisfies
$$\alpha u - \lambda \bM v = \beta u \quad \text{and} \quad -\lambda \bM^\top u + v = \beta v$$
which implies $(1-\beta)v = \lambda \bM^\top u$ and so $\lambda^2 \bM \bM^\top u = (\alpha-\beta)(1-\beta)u$. This means either $u$ is an eigenvector of $\lambda^2 \bM \bM^\top$ with eigenvalue $(\alpha-\beta)(1-\beta)$, or $u=0$ and $\beta \in \{1,\alpha\}$. Conversely, if $u$ is an eigenvector of $\lambda^2 \bM \bM^\top$ with eigenvalue $(\alpha-\beta)(1-\beta) \ne 0$, then $[u \; v]^\top$ with $v = (1-\beta)^{-1} \lambda \bM^\top u$ is an eigenvector of $\widetilde \bH^{\le 2}$ with eigenvalue $\beta$. Letting $\mu_1 > \cdots > \mu_n > 0$ be the eigenvalues of $\lambda^2 \bM \bM^\top$, $\widetilde \bH^{\le 2}$ has $2n$ eigenvalues of the form
$$\frac{\alpha+1 \pm \sqrt{(\alpha-1)^2 + 4 \mu_i}}{2}$$
and the remaining eigenvalues are $\alpha$ or $1$. Thus, the $u$-part of the bottom eigenvector of $\widetilde \bH^{\le 2}$ is precisely the leading eigenvector of $\bM \bM^\top$. This is a close variant of the tensor unfolding spectral method (see Section~\ref{sec:unfolding}), and we expect that its performance is essentially identical.

\subsection{The General Case: $r \ge p$}

One difficulty when $r > p$ is that there is no longer a trivial stationary point that we can write down in closed form. There is, however, a natural guess for ``uninformative'' beliefs that only depend on the local information: for $0 < |S| \le r$,
$$b_S(x_S) = \frac{1}{Z_S} \exp\left(\lambda \sum_{\substack{U \subseteq S \\ |U| = p}} \bY_U x^U\right)$$
for the appropriate normalizing factor $Z_S$. Unfortunately, these beliefs are not quite consistent, so we need separate moments for each set $S$:
$$m^{(S)}_T = \Ex_{x_S \sim b_S}[x^T].$$

\noindent Provided $\lambda \ll 1$, we can check that $m^{(S)}_T \approx m^{(S')}_T$ to first order, and so the above beliefs are at least approximately consistent:
$$Z_S = \sum_{x_S} \exp\left(\lambda \sum_{\substack{U \subseteq S \\ |U| = p}} \bY_U x^U\right) \approx \sum_{x_S} \left(1 + \lambda \sum_{\substack{U \subseteq S \\ |U| = p}} \bY_U x^U\right) = 2^{|S|}$$
and so
\begin{align*}
m_T^{(S)} &= \sum_{x_S} b_S(x_S) x^T\\
&= \frac{1}{Z_S} \sum_{x_S} x^T \exp\left(\lambda \sum_{\substack{U \subseteq S \\ |U| = p}} \bY_U x^U\right)\\
&\approx \frac{1}{Z_S} \sum_{x_S} x^T \left(1 + \lambda \sum_{\substack{U \subseteq S \\ |U| = p}} \bY_U x^U\right)\\
&= \left\{\begin{array}{ll} \lambda \bY_T & \text{if } |T|=p \\ 0 & \text{otherwise}\end{array}\right.
\end{align*}
which does not depend on $S$. Thus we will ignore the slight inconsistencies and carry on with the derivation. As above, the important calculation is, for $T,T' \subseteq S$,
\begin{align*}
\frac{\partial^2 \mathcal{S}_S}{\partial m_T^{(S)} \partial m_{T'}^{(S)}} &= -2^{-2|S|} \sum_{x_S} x^{T \symd T'} b_S(x_S)^{-1}\\
&= -2^{-2|S|} \sum_{x_S} x^{T \symd T'} Z_S \exp\left(-\lambda \sum_{\substack{U \subseteq S \\ |U| = p}} \bY_U x^U\right)\\
&\approx -2^{-|S|} \sum_{x_S} x^{T \symd T'} \exp\left(-\lambda \sum_{\substack{U \subseteq S \\ |U| = p}} \bY_U x^U\right)\\
&\approx -2^{-|S|} \sum_{x_S} x^{T \symd T'} \left(1-\lambda \sum_{\substack{U \subseteq S \\ |U| = p}} \bY_U x^U\right)\\
&= \left\{\begin{array}{ll}-1 & \text{if } T=T' \\ \lambda \bY_{T \symd T'} & \text{if } |T \symd T'| = p \\ 0 & \text{otherwise.}\end{array}\right.
\end{align*}

\noindent Analogous to \eqref{eq:H-sum}, we compute the Kikuchi Hessian
$$\bH_{T,T'} \defeq -\sum_{\substack{S \supseteq T \cup T' \\ |S| \le r}} c_S \frac{\partial^2 \mathcal{S}_S}{\partial m_T^{(S)} \partial m_{T'}^{(S)}}.$$

\noindent If we fix $\ell_1, \ell_2$, the submatrix $\bH^{(\ell_1,\ell_2)} = (\bH_{T,T'})_{|T| = \ell_1, |T'| = \ell_2}$ takes the form
$$\bH^{(\ell_1,\ell_2)} \approx a(\ell_1,\ell_2) \bone_{\ell_1 = \ell_2} \bI - b(\ell_1,\ell_2) \bM^{(\ell_1,\ell_2)}$$
for certain scalars $a(\ell_1,\ell_2)$ and $b(\ell_1,\ell_2)$, where $\bI$ is the identity matrix and $\bM^{(\ell_1,\ell_2)} \in \RR^{{[n] \choose \ell_1} \times {[n] \choose \ell_2}}$ is the symmetric difference matrix
$$\bM^{(\ell_1,\ell_2)}_{T,T'} = \left\{\begin{array}{ll} \bY_{T \symd T'} & \text{if } |T \symd T'| = p \\ 0 & \text{otherwise.}\end{array}\right.$$

Instead of working with the entire Kikuchi Hessian, we choose to work instead with $\bM^{(\ell,\ell)}$, which (when $p$ is even) appears as a diagonal block of the Kikuchi Hessian whenever $r \ge \ell + p/2$ (since there must exist $|T| = |T'| = \ell$ with $|T \cup T'| \le r$ and $|T \symd T'| = p$). Our theoretical results (see Section~\ref{sec:main-results}) show that indeed $\bM^{(\ell,\ell)}$ yields algorithms matching the (conjectured optimal) performance of sum-of-squares. When $p$ is odd, $\bM^{(\ell,\ell)} = 0$ and so we propose to instead focus on $\bM^{(\ell,\ell+1)}$; see Appendix~\ref{app:odd}.

\end{document}